\documentclass[11pt]{article}

\usepackage{amsmath,amsfonts,amsthm,amssymb,color,bbm}
\usepackage{fancybox}
\usepackage{algorithm}
\usepackage[noend]{algpseudocode}
\usepackage{dsfont}
\usepackage{bm}
\usepackage{hyperref,cleveref}
\usepackage{subfig}
\usepackage[margin=1in]{geometry}

\newcommand{\set}[1]{\left\{#1\right\}}

\newtheorem{theorem}{Theorem}
\newtheorem{lemma}[theorem]{Lemma}
\newtheorem{claim}[theorem]{Claim}
\newtheorem{fact}[theorem]{Fact}

  {\begin{Sbox}\begin{minipage}}%
  {\end{minipage}\end{Sbox}\fbox{\TheSbox}}

\def\expec#1#2{{\mathbb{E}}_{#1}\left[ #2 \right]}
\newcommand{\E}{\mbox{{\bf E}}}

\def\defeq{\stackrel{\mathrm{def}}{=}}

\def\R{\mathbb{R}}

\usepackage{pgf,tikz}
\usetikzlibrary{arrows}

\newcommand{\beq}{\begin{equation}}
\newcommand{\eeq}{\end{equation}}

\renewcommand{\P}{{\mathcal{F}}}

\usepackage{xcolor,cite,etoolbox}
\makeatletter 
\pretocmd\@bibitem{\color{black}\csname keycolor#1\endcsname}{}{\fail}
\newcommand\citecolor[1]{\@namedef{keycolor#1}{\color{blue}}}
\makeatother
\citecolor{ChekuriM97}

\usepackage{times}
\usepackage[makeroom]{cancel}

\newcommand{\cov}{\text{Cov}}

\begin{document}

\title{
Improved Approximations for Min Sum Vertex Cover and Generalized Min Sum Set Cover
}
\date{}
\author{
Nikhil Bansal\thanks{CWI and TU Eindhoven, the Netherlands.
\texttt{bansal@gmail.com}. Supported by a NWO Vici grant 639.023.812.}
\and
Jatin Batra\thanks{CWI, the Netherlands.
\texttt{jatinbatra50@gmail.com}. Supported by a NWO Vici grant 639.023.812.}
\and
Majid Farhadi\thanks{Georgia Institute of Technology.
\texttt{farhadi@gatech.edu}.
Supported in part by the ACO Ph.D. Program.}
\and
Prasad Tetali\thanks{Georgia Institute of Technology.
\texttt{tetali@math.gatech.edu}.
Supported in part by the NSF grants DMS-1811935 and NSF TRIPODS-1740776.}
}


\maketitle

\begin{abstract}
We study the \emph{generalized min sum set cover} (GMSSC) problem, wherein given a collection of hyperedges $E$ with arbitrary covering requirements $\set{k_e\in Z^+: e \in E}$, the goal is to find an ordering of the vertices to minimize the total cover time of the hyperedges; a hyperedge $e$ is considered covered by the first time when $k_e$ many of its vertices appear in the ordering.

\medskip

We give a $4.642$ approximation algorithm for GMSSC, coming close to the best possible bound of $4$, already for the classical special case (with all $k_e=1$) of \emph{min sum set cover} (MSSC)  
studied by Feige, Lov\'{a}sz and Tetali \cite{FLT04}, and improving upon the previous best known bound of $12.4$  due to Im, Sviridenko and van der Zwaan \cite{ISV14}. Our algorithm is based on transforming the LP solution by a suitable kernel and applying randomized rounding. This also gives an LP-based $4$ approximation for MSSC. As part of the analysis of our algorithm, we also derive an  inequality on the lower tail of a sum of independent Bernoulli random variables, which might be of independent interest and broader utility.

\medskip

Another well-known special case is the \emph{min sum vertex cover} (MSVC) problem, in which the input hypergraph is a graph (i.e., $|e| = 2$) and $k_e = 1$, for every edge $e \in E$. We give a $16/9 \simeq 1.778$ approximation for MSVC, and show a matching integrality gap for the natural LP relaxation. This improves upon the previous best $1.999946$ approximation of Barenholz, Feige and Peleg \cite{BFP06}. (The claimed $1.79$ approximation result of Iwata, Tetali and Tripathi \cite{ITT12} for the MSVC turned out have an unfortunate, seemingly unfixable, mistake in it.)

\medskip

Finally, we revisit MSSC and consider the $\ell_p$ norm of cover-time of the hyperedges.
Using a dual fitting argument, 
we show that the natural greedy algorithm achieves tight, up to NP-hardness, approximation guarantees of $(p+1)^{1+1/p}$, for all $p\ge 1$.
For $p=1$, this gives yet another proof of the $4$ approximation for MSSC.

\end{abstract}


\section{Introduction}
In the {\em min sum set cover problem (MSSC)}, formally introduced by Feige, Lov\'asz and Tetali \cite{FLT04}, given a collection of sets whose union is $V$, one seeks an ordering of the elements of $V$ so as to minimize the sum of the cover times of the sets. Feige et al.\ showed that a natural greedy algorithm provides a factor 4 approximation (a result implicit in the prior work of Bar-Noy et al.\ \cite{BBHST98}, albeit with a more complicated proof), and also established that it is hard to approximate it to within a ratio $4-\epsilon$, for every $\epsilon>0$.

Azar, Gamzu and Yin~\cite{AGY09} introduced a significant generalization of MSSC as the multiple intents re-ranking or the \textbf{generalized min sum set cover (GMSSC)} problem. In GMSSC, the input consists of a hypergraph $H = (V, E)$ where every hyperedge has a covering requirement $k_e$ where $k_e \leq |e|$. The output is an ordering of the vertices, i.e., an assignment of $V$ to time slots $1,\ldots,n$. We say that $e$ is covered at time $t$, if $t$ is the earliest time by which $k_e$ of the vertices contained in $e$ have appeared in the ordering. The objective is to minimize the total cover-time of all the hyperedges. Note that the case of $k_e=1$, for all $e\in E$, corresponds to MSSC. GMSSC has applications in query results diversification and broadcast scheduling among others\cite{AGY09,tsaparas2011selecting,kofler2014intent,kofler2015user}.

Azar et al.\ provided an $O(\log r)$ approximation for the problem, where $r=\max_e k_e$. Bansal, Gupta, and Krishnaswamy~\cite{BGK10} were the first to provide a constant-factor (of $485$) approximation for GMSSC by
introducing a strong LP formulation using the so-called Knapsack Cover (KC) inequalities.
Their approximation bound was  improved to $28$, by Skutella and Williamson~\cite{SW11}, using the idea of $\alpha$-point rounding and other technical enhancements. This has since been improved to 12.4  by Im, Sviridenko and van der Zwaan \cite{ISV14}. They introduce a different configuration LP, and show how to obtain a {\em preemptive} schedule from this LP losing a factor $2$, and then round it to obtain the final solution at another factor $6.2$ loss. In the quest for the best possible $4$ approximation, they conjecture that any preemptive schedule can be rounded to an integral solution at factor $2$ loss.

\textbf{Min sum vertex cover (MSVC)} is a well-known special case of MSSC, in which the hypergraph is a graph,
i.e. $|e| = 2$ and $k_e = 1$,  $\forall e \in E$.
MSVC first arose in optimizing matrix computations involved in the analysis of a heuristic to speed up semidefinite program solvers~\cite{BM01}. Feige et al.~\cite{FLT04}, provided a randomized $2$-approximation algorithm for MSVC, which was later  improved to $1.999946$ by Barenholz, Feige, and Peleg~\cite{BFP06}. It should be noted that a 
1.79-approximation result for MSVC, reported in \cite{ITT12}, turned out to have  a crucial error in the proof of Lemma~1 of \cite{ITT12} according to its authors, thus invalidating the claimed 1.79 bound.

\subsection{Our results and techniques}

Despite the only lower bound on the approximability of GMSSC being 4 for the special case of MSSC by Feige et al., the case of arbitrary covering requirements has remained difficult to approximate within a factor close to $4$. We obtain the following result for GMSSC.

\begin{theorem}\label{thm:gmsscintro}
There is a polynomial time $4.642$-approximation algorithm for the generalized min sum set cover (GMSSC) problem.
\end{theorem}

Numerical evidence suggests that the approximation ratio of our algorithm is no more than $4.5232$. We also note that the approach of~\cite{ISV14} cannot give an approximation better than $8$. In particular, we give a counter-example to their conjecture on the gap between a preemptive schedule and a non-preemptive one.

All known approaches for GMSSC start with obtaining a fractional solution to some strengthening of the natural LP (which has unbounded integrality gap~\cite{BGK10}). Applying a simple randomized rounding however, fails even for $k_e=1$ as the LP may cover a hyperedge $e$ at time $t$ by scheduling $r$ of its vertices to extent $1/r$, while randomized rounding will leave it uncovered with probability $(1-1/r)^r \approx 1/e$  for times much greater than $t$. Therefore, all approaches first modify the LP solution in some way, to increase the probability of a hyperedge being covered early relative to its LP cost.

\noindent\textbf{Our approach.} The key to our results is to use a careful linear transformation to perform such a modification. Given a fractional assignment $x_{v,t}$ of vertices to time-slots, we obtain a new solution $z_{v,t}$  by applying the linear transformation $K$, referred to henceforth as the \emph{kernel},
\begin{equation} 
\label{eq:kernelintro} z_{v,t} =\sum_{t'} K(t,t')  x_{v,t'}\,.
\end{equation}
We then apply a standard randomized rounding ($\alpha$-point rounding) to $z$.
To prove Theorem~\ref{thm:gmsscintro}, we use the kernel $K(t,t') = \beta/t$ (for $t \geq t'$), for a suitable constant $\beta$. 
To see the idea, suppose for some $v$, $x_{v,1} > 0$ and $x_{v,t} = 0$ for $t>1$. Then, $z_{v,t} \approx \beta x_{v,1}/t$. Roughly speaking, the solution $x$ is ``spread" to the right at a rate $\beta/t$, and eventually the cumulative amount of $v$ scheduled by $z$ will be arbitrarily large.

Choosing $\beta=2$ gives a new LP rounding based $4$-approximation for MSSC, refining the result in \cite{FLT04}, as the guarantee there was with respect to the integral optimum solution.
\begin{theorem}\label{thm:msscintro}
There is an LP-rounding based $4$-approximation algorithm for MSSC.
\end{theorem}

\noindent{\bf Min sum vertex cover.} For MSVC, despite significant effort to improve the factor 2 approximation of~\cite{FLT04}, no approximation guarantee substantially better than 2 is known 
\cite{BFP06}. We show the following.
\begin{theorem}\label{thm:msvcintro}
There is a polynomial time $16/9$ approximation for MSVC. Further, this is tight with respect to the natural LP relaxation.
\end{theorem}
A key idea is to use a different kernel for MSVC based on the following observation: for any hyperedge $e$, the LP must schedule some vertex $v \in e$ to an extent of at least $1/2$ (in contrast to $1/r$ for MSSC above). As it suffices for randomized rounding that  $v$ is scheduled by $z$ to extent $1$, we can 
spread $x$ much less aggressively than for MSSC. We show that using the kernel $K(t,t') \approx \beta (t')^2/t^3$ (for $t \geq t'$) gives us Theorem~\ref{thm:msvcintro}.

\noindent\textbf{$\ell_p$ norms.} Finally we consider the $\ell_p$ norms of cover times for MSSC. Using the kernel $K(t,t') = (p+1)/t$  (for $t \geq t'$) directly gives a $\delta_p:=(p+1)^{1+1/p}$ approximation guarantee, which is also the best possible as shown in Theorem \ref{thm:gr} below. 
However, we also give a different perspective by showing that the greedy algorithm of~\cite{FLT04} (which is oblivious to $p$), gives a $\delta_p$ approximation (simultaneously) for every $p$.
\begin{theorem}\label{thm:gr}
For any $p \in [1,\infty)$ the greedy algorithm guarantees a $\delta_p$ approximate solution for the min $\ell_p$ norm set cover problem. Further, it is NP-hard to approximate the min $\ell_p$ norm set cover problem better than $\delta_p - \epsilon$ for any $\epsilon > 0$ and any $p\geq 1$.
\end{theorem}
The proof of Theorem \ref{thm:gr} is based on a \emph{dual fitting} argument. In particular, for $p=1$, this provides yet another LP based $4$-approximation for MSSC.

\noindent{\bf Analysis.}
While our algorithms are quite simple, the analysis is more subtle.
A key idea which allows us to get tight or close to tight bounds in Theorems \ref{thm:gmsscintro}-\ref{thm:msvcintro}, is the setting up of a non-linear convex optimization problem to find the worst-case fractional solution which maximizes our approximation ratio. 

For GMSSC, the solution to the optimization problem above, is an expression in terms of the following function $P(\gamma)$, involving the sum of Bernoulli random variables:
Given independent Bernoulli random variables ${Y_v}$ with $\expec{}{\sum_v Y_v} = \gamma k$, where $k$ is a positive integer and $\gamma \geq 1$. What is the best
upper bound $P(\gamma)$ on $\Pr[\sum_v Y_v \leq k-1]$, as a function of $\gamma$?

Standard Chernoff-Hoeffding or other tail bounds are inadequate for our purposes as they lose relatively large constants. 
We use the following refined bound which may be of independent interest.
\begin{theorem}\label{thm:enhanced-bound}
Let $S$ be the sum of $n$ independent, not necessarily identical, Bernoulli random variables. For any integer $k \leq \expec{}{S}$, letting $\gamma = \expec{}{S} / k \ge 1$, we have
\[ \Pr[S < k] \leq e^{-\gamma} (e/2)^{e^{1-\gamma}}.\]
\end{theorem}
In particular, for $\gamma=1$, $P(\gamma)=1/2$ which is the best possible as $\Pr[S \leq k-1]$ approaches $1/2$ for large $k$ and $\expec{}{S} =  k$.
On the other extreme, for large $\gamma$, $P(\gamma)$  approaches $e^{-\gamma}$. This again is the best possible as for $k=1$, $\Pr[S \leq k-1] = \Pr[ S =0] \approx \exp(-\gamma)$. 

The bound above can probably be strengthened. Numerical evidence suggests that
the stronger bound $Pr[S < k] \leq e^{-\gamma} (e/2)^{e^{c(1-\gamma)}}$ with $c=13.3$ holds for all $k$.
If true, this would imply an improved bound of $4.5232$ for GMSSC in Theorem \ref{thm:gmsscintro}  as mentioned previously. 
However, we remark that this is not the only reason that our bound for GMSSC is not $4$.
Even if we assume the best possible function $P(\gamma) = e^{-\gamma}$, the best achievable bound in Theorem \ref{thm:gmsscintro} using our approach is $4.509$. 
This is because unlike in the argument for MSSC, the use of Knapsack Cover inequalities 
in the argument for GMSSC leads to an additional loss.

\subsection{Other Related Works}
A well-studied special case of GMSSC is when all $k_e =|e|$, a  problem known as the \emph{min-latency set cover}, for which Hassin and Levin \cite{HL05} gave an  $e$-approximation. However, this is equivalent to the classical problem of minimizing weighted completion time with precedence constraints,  $1|\text{prec}|\sum_j w_j C_j$, \cite{W03}, for which several $2$-approximations are known \cite{hall1997scheduling,chekuri1999precedence,margot2003decompositions}. A matching hardness of $2 - \epsilon$ (for every $\epsilon>0$) is also known, assuming a variant of the Unique Games Conjecture \cite{BK09}.
Close to this end also lies the \emph{all-but-one MSSC}, another special case of GMSSC, with $k_e = \max \set{1, |e|-1} \forall e\in E$, for which Happach and Schulz \cite{HS20} provided a $4$-approximation.

MSSC is closely related to the min-sum graph coloring problem \cite{BBHST98} and has found interesting applications in peer to peer networks \cite{CFK03}, data streams \cite{BMMNW04}, and data-base query processing optimization \cite{MBMW05}. It has inspired more general models and their analyses, e.g., pipelined set cover problem \cite{MBMW05} and versions with
precedence-constraints \cite{MMW17,HS20}.
Finally variants of MSSC under submodular and supermodular cost functions have also been studied~\cite{P92,ITT12,FLV19,HHL20}.

\subsection{Organization}
The rest of the paper is organized as follows. 

In Section~\ref{sec:framework}, we describe the LP formulation for GMSSC based on the Knapsack Cover inequalities, and the template based on kernels and $\alpha$-point rounding that we use in all our algorithms. In Section~\ref{sec:GMSSC}, we consider the GMSSC problem. We begin with the simpler MSSC and min latency problems, describing a $4$ and an $e$ approximation for them respectively using our analysis framework.  Building on these ideas we then present our main result for GMSSC in \S\ref{sec:gmssc}.
The lower tail for sums of Bernoulli r.v.s is proved in Section~\ref{sec:tail} and Appendix~\ref{sec:enhanced}. 

In Section~\ref{sec:MSVC}, we consider the MSVC problem and prove the $16/9$ approximation, and also show the matching integrality gap for the LP.
In Section~\ref{sec:minLPsetcover}, we consider the problem of minimization of the $\ell_p$ norms of cover times for MSSC.
Finally, in Section~\ref{sec:imconj}, we describe the counter-example for the conjecture of \cite{ISV14}.

\section{The Basic Framework}\label{sec:framework}
We assume that time is slotted, and time $t$, for $t=1,2,\ldots$, refers to the time interval $(t-1,t]$. We refer to edges and hyperedges interchangeably.

\noindent\textbf{LP relaxation for MSSC and GMSSC.}
Consider the following natural formulation for MSSC. For each vertex $v$ and time $t$ we have a variable $x_{v,t}$, which is intended to be $1$ if $v$ is assigned to time $t$. 
For each edge $e$ and time $t$, there is a variable $u_{e,t}$ which is intended to be $1$ if $e$ is not covered at the start of $t$.
Subsequently, we have the following LP relaxation.
\begin{align}
\text{Minimize \quad\quad\quad}\sum\limits_{e,t}   u_{e,t}  \qquad \text{Subject to} \qquad u_{e,t}, x_{v,t} &\ge 0, \qquad \forall~ e,v,t \qquad \qquad \qquad \qquad \qquad \nonumber \\ 
    \sum\limits_{v}  x_{v,t} & \leq 1,   \qquad\forall~ t\label{eq:lppack}\\
   u_{e,t} + \sum_{v \in e} \sum_{t'<t} x_{v,t^{\prime}} & \ge 1,   \qquad\forall~ e,t \label{eq:lpcovbasic}
\end{align}
Constraint~\eqref{eq:lppack} ensures that at most one vertex is assigned to any time, and constraint~\eqref{eq:lpcovbasic} ensures that $u_{e,t}$ is $0$ only if some $v \in e$ is scheduled strictly before $t$.
Note that the objective can also be written equivalently as $\sum_t t ( u_{e,t}-u_{e,t+1})$ as $u_{e,t}-u_{e,t+1}$ can be viewed as the amount of $e$ fractionally completed at time $t$. Later on, we will use both these perspectives interchangeably.

For GMSSC, where the demands $k_e$ are arbitrary, the natural extension of this LP becomes extremely weak, but as shown in \cite{BGK10}, it can be strengthened by replacing \eqref{eq:lpcovbasic} by the following Knapsack Cover (KC) inequalities:
\begin{equation}  (k_e-|S|) u_{e,t} + \sum_{v \in e\setminus S}  \sum_{t^{\prime}< t} x_{v,t^{\prime}}  \ge (k_e-|S|), \qquad\forall~ e,t,S \subseteq e, |S| < k_e\,.  \label{eq:lpcov}
\end{equation}
The constraints \eqref{eq:lpcov} require that for an edge $e$ to be considered covered at time $t$, no matter which subset $S$ of vertices in $e$ we ignore,  at least $(k_e - |S|)$ many  must be scheduled from vertices in $e\setminus S$ before $t$.  
Henceforth, we will treat the LP with constraints \eqref{eq:lppack} and \eqref{eq:lpcov} as the LP for GMSSC, unless otherwise specified.
Even though it has exponentially many constraints, it can be solved efficiently to any desired accuracy, see~\cite{BGK10} for details.

\noindent\textbf{The algorithmic approach.}
All our algorithms follow the same basic template consisting of three steps:

First, compute some optimum solution $x$ for the corresponding LP relaxation.
For a vertex $v$, let $x_v = (x_{v,t})_{\set{t}}$ denote the vector corresponding to its fractional assignment.

Second, for each vertex $v$, apply a suitable (problem dependent) linear transformation $K$ to $x_v$ to obtain $z_v$. That is, for all $v$ and all $t$,
\begin{equation} 
\label{eq:kernel} z_{v,t} =\sum_{t'} K(t,t')  x_{v,t'}\,.
\end{equation}
We will refer to the matrix $K$ as a kernel.
The solution $z$ may not satisfy \eqref{eq:lpcov} and \eqref{eq:lppack}  anymore, but $K$ will be chosen so that $z$ satisfies other useful properties.

Third, create a random tentative schedule $\tau$ (where multiple vertices may be assigned to the same time slot), by applying $\alpha$-point rounding independently for each vertex $v$, using the solution $z_v$ (see Algorithm \ref{meta-algorithm} below). 
Next, convert the schedule $\tau$ to a proper schedule $\sigma$ by scheduling one vertex at a time in the order given by $\tau$, breaking ties at random (this randomness will be crucial for some of the problems we consider). 

\makeatletter
\def\BState{\State\hskip-\ALG@thistlm}
\makeatother

\begin{algorithm}
\caption{The algorithmic template}\label{meta-algorithm}
\begin{algorithmic}[1]
\Procedure{Kernel $\alpha$-point Rounding}{$H = (V, E), k: E \rightarrow {\mathbb{N}}$}
\State $x \gets \text{An optimum fractional schedule for $H, k$.}$ \Comment{i: Solving a Relaxation.}
\State $z_v \gets K x_v \quad \forall v \in V$. \Comment{ii: Applying the Kernel}
\For{$v \in V$} \Comment{iii: $\alpha$-point Rounding} 
        \State $\alpha_v \sim \textit{uniform}[0,1]$. 
        \State $\tau_{v,t} \gets \mathbbm{1}[ t = $ the earliest time for which $\sum_{t^{\prime} \leq t} z_{v,t^{\prime}} \ge \alpha_v] \quad \forall t$.  \Comment{Tentative Schedule}
\EndFor
\State \Return an ordering $\sigma$ of $V$ :  Scheduling vertices according to $\tau$, breaking ties at random.
\EndProcedure
\end{algorithmic}
\end{algorithm}

\noindent\textbf{Analysis.}
The analyses of our algorithms will also follow the same template.
We will do a per-edge analysis. We first introduce some notation that will be used throughout. For an edge $e$, let $c_x(e) = \sum_{t} u_{e,t}$ denote the cost for $e$ in the LP solution $x$.

Given $z$, and the random choices $\alpha_v$, the cover time $\cov_\tau(e)$ of $e$ in the (random) tentative schedule $\tau$ is the earliest time $t$, by which exactly $k(e)$ elements of $e$ are scheduled in $\tau$. Note that the only randomness in $\cov_\tau(e)$ is due to the choice of $\alpha_v$, for all $v \in e$.

Similarly, let $\cov_\sigma(e)$ be the (random) cover time of $e$ in the proper schedule $\sigma$, obtained from $\tau$. The randomness in $\cov_\sigma(e)$ is due to choices of $\alpha_v$ for all $v$ (including vertices not in $e$, as these can {\em delay} vertices in $e$), and the randomness of the tie-breaking rule at each slot.

To implement the aforementioned, first we define a cost $c_z(e)$ that upper bounds $\E[\cov_\tau(e)]$. Next, we 
bound the ratio of $c_z(e)$  to $c_x(e)$ and finally, we bound the expected cover time $\E[\cov_\sigma(e)]$ under $\sigma$ 
by a multiple of $c_z(e)$.
Relating $c_z(e)$ to $c_x(e)$ will be the hardest part. For each of the problems we consider, we will bound this ratio by considering another optimization problem to solve for the worst-case setting of the variables $x_{v,t}$ that maximizes this ratio and bounding its optimum.

\section{Generalized Min Sum Set Cover}\label{sec:GMSSC}
In this section we develop the algorithm for GMSSC. 
To motivate the eventual analysis of the algorithm, we first start with the simpler MSSC problem and describe a tight factor $4$ approximation. Then, we describe a factor $e$ approximation for the min-latency version. Finally, we describe our algorithm for GMSSC in \S\ref{sec:gmssc}.

For all these problems we apply Algorithm~\ref{meta-algorithm}, and we use a kernel of the form $K(t,t') = \frac{\beta}{t}$
for $t'\leq t$, and $K(t,t')=0$ otherwise. Here $\beta$ is a constant whose choice will be optimized separately for each problem. 

Given an optimum solution $x$ to the  underlying LP, we obtain $z$ by applying $K$, so for each $v,t$, we have     
\[z_{v,t} = \frac{\beta}{t} \sum_{t^{\prime}\leq t} x_{v,t^{\prime}}\,.\]
This kernel has the interesting property that even if $v$ is scheduled in the LP solution $x$ to extent at most $\epsilon$, the amount of $v$ in $z$ will be arbitrarily large. E.g.~suppose $x_{v,1}=\epsilon$ and $0$ otherwise, then $z_{v,t} =\beta \epsilon/t$ and hence $\sum_t z_{v,t}$ diverges. 
Yet, crucially, the total {\em load} at any time step in $z$ is bounded. In particular, observe:
\begin{claim}\label{cl:widthz}
For any $t$, $\sum_{v} z_{v,t} \leq \beta$.
\end{claim}
\begin{proof}
As \eqref{eq:lppack} ensures that $\sum_{v} x_{v,t'} \leq 1$ for any $t'$, we have
$\sum_v z_v = \sum_v Kx_v \leq K \mathbf{1} = \beta\mathbf{1}.$
\end{proof}

For any vector $b=(b_1,b_2,\ldots)$, let $b_{<t}=b_1+\ldots+b_{t-1}$ and $b_{\leq t} = b_1+\ldots+b_t$.
We call a vector non-negative if all its entries are non-negative. For two non-negative vectors $a,a'$, we denote $a \succeq a'$ if $a_{\leq t} \geq a'_{\leq t}$ for all $t$.
For a real number $c$, $\max(c,0)$ is denoted by $(c)_+$.

We will use the following repeatedly.
\begin{claim}\label{cl:kernelog}
Let $a$ be a non-negative vector and let $b = K a$. Then, 
\[b_{<t} = \sum_{t' < t} b_{t'} \geq \beta \sum_{t' \leq t} a_{t'} \ln \frac{t}{t'}.\]
Moreover, if $a\succeq a'$, then $K a \succeq K a'$.
\end{claim}
\begin{proof}
By the definition of K, and using that $\sum^{d-1}_{i=c} \frac{1}{i} \geq \ln(\frac{d}{c})$, for positive integers $c,d$,  we have that
\[
b_{< t} = \sum_{t''< t} b_{t''}  =\sum_{t''< t} \sum_{t'\leq t''} K(t'',t')a_{t'} = \sum_{t' \leq t} a_{t'} \sum_{t' \leq t'' < t} K(t'',t) \geq  \beta\sum_{t' < t} a_{t'} \ln{\frac{t}{t'}} = \beta\sum_{t' \leq t} a_{t'} \ln{\frac{t}{t'}}.\]
As entries of $K$ are non-negative, 
if $a_{\leq t} \geq a'_{\leq t}$,  then $(Ka)_{<t} \geq (Ka')_{<t}$ for all $t$, and hence $Ka \succeq K a'$.
\end{proof}

Let us state another simple fact that we will use several times later on.
\begin{fact} \label{fact}
Let $f:\R^n \rightarrow \R$ be a convex function that is non-negative on some compact region $H$, and let $g: \R^n \rightarrow \R$ be a linear function that is strictly positive on $H$, then the maximum of $\max_{x\in H} f(x)/g(x)$ is attained at an extreme point of $H$.
\end{fact}
\begin{proof}
This follows by observing that for any $x,y \in H$, and $\lambda \in [0,1]$, as $f,g$ are non-negative on $H$,  
\[ \max \left( \frac{f(x)}{g(x)}, \frac{f(y)}{g(y)}\right) \geq \frac{ \lambda f(x) + (1-\lambda) f(y) }{\lambda g(x) + (1-\lambda) g(y) } \geq \frac { f(\lambda x + (1-\lambda)y)}{g(\lambda x + (1-\lambda) y)}\,,\]
where the last step uses the convexity of $f$ and linearity of $g$.
\end{proof}
\subsection{Min Sum Set Cover}
\label{s:mssc}
Here we show the following result for MSSC, refining the previous approximation of \cite{FLT04}, which was with respect to the integral optimum solution.
\begin{theorem}
\label{thm:mssc}
For any $\beta \geq 1$, the algorithm described above is a $\beta^2/(\beta-1)$-approximation. In particular, setting $\beta=2$ gives a $4$-approximation.
\end{theorem}

Before proving this theorem, we first give some notation. Recall that $\cov_\tau(e)$ is the cover time of $e$ in the tentative schedule $\tau$, and $\cov_\sigma(e)$ in the proper schedule $\sigma$. Moreover, $c_x(e)$ is the LP cost for $e$.  

By the definition of the schedule $\tau$, note that $\E[\cov_\tau(e)] = \sum_{t} p_t(e)$, where
 $p_t(e)$ is the probability that edge $e$ is uncovered at the beginning of time-slot $t$ in $\tau$.
 As each $v$ is rounded independently using $\alpha$-point rounding,
 \beq \label{eq:pt} p_t(e) = \prod_{v \in e} (1-z_{v,<t})_+   \leq \exp(-\sum_{v \in e} z_{v,<t}). \eeq
 Let us define
\[c_z(e) = \sum_t \exp(-\sum_{v\in e} z_{v,<t}),\]
and note that it upper bounds $\expec{}{\cov_\tau(e)}$.
Theorem~\ref{thm:mssc} is a direct consequence of the next two lemmas. 

\begin{lemma} 
\label{lem:tentcostmssc} For each edge $e$, \ $c_z(e)\leq \frac{\beta}{\beta-1} c_x(e)$.
\end{lemma}
\begin{lemma}\label{lem:msscfin} For each edge $e$,\  $\expec{}{\cov_\sigma(e)} \leq \beta c_z(e)$.
\end{lemma}
Intuitively, Lemma~\ref{lem:msscfin} follows from what we observed as Claim~\ref{cl:widthz}: As $\tau$ has at most $\beta$ vertices in expectation at any time $t$,  the expected cover time of $e$ in $\sigma$ should be at most $\beta$ times that in $\tau$, so can be upper bounded by $\beta c_z(e)$. However, to make this intuition precise, one needs to {\em condition} on when some vertex of $e$ was first scheduled in $\tau$.

We address this in \S\ref{sec:conditioning}, and focus here on Lemma~\ref{lem:tentcostmssc} instead.

 \begin{proof}[Proof of Lemma \ref{lem:tentcostmssc}]
 As $z_v = K x_v$, by Claim \ref{cl:kernelog}, we have that $z_{v,<t} \geq \beta\sum_{t' \leq t} x_{v,t'} \ln(t/t')$, and hence
 \[ c_z(e) =  \sum_t \exp(-z_{v,<t}) \leq \sum_t \exp \Big(-\beta\sum_{v \in e}\sum_{t^{\prime}\leq t} x_{v,t^{\prime}} \ln (t/t^{\prime}) \Big).\]
 
 Recalling that $c_x(e) = \sum_t u_{e,t} = \sum_t (1-  \sum_{v \in e} x_{v,<t})_+$, to prove the lemma, it suffices to show that 
\beq
\label{lem:ptbound}
\sum_t \exp \Big(-\beta\sum_{t^{\prime}\leq t}\sum_{v \in e} x_{v,t^{\prime}} \ln (t/t^{\prime}) \Big) \leq \frac{\beta}{(\beta-1)}   \sum_t (1-  \sum_{v \in e} x_{v,<t})_+.
\eeq
holds for any setting of $x_{v,t}$ variables.

Let us define $x_{e,t} = u_{e,t} - u_{e,t+1}$ (the amount of $e$ satisfied at time $t$ in solution $x$). Noting that $x_{e,t} \leq \sum_{v \in e} x_{v,t}$, replacing $\sum_{v \in e}x_{v,t'}$ by $x_{e,t'}$ can only increase the left hand side of~\eqref{lem:ptbound}.

 Moreover, $x_{e,<T} = 1$ for large enough $T$, as $e$ will eventually be satisfied by the LP.  So, the cost $c_x(e)$ can be written as $\sum_t t x_{e,t}$ and our goal is to show that for any $x_{e,t}\geq 0$ satisfying $\sum_t x_{e,t} =1$, the following inequality holds:
\[ \sum_t \exp \Big(-\beta \sum_{t^{\prime}\leq t}  x_{e,t'} \ln(t/t^{\prime}) \Big) \leq \frac{\beta}{(\beta-1)} \sum_t t x_{e,t}\,. \]
\noindent
To this end, we consider the following optimization problem in  the variables $a_t$.
\[ (\P) \qquad  \label{def:probmssc} \text{Maximize }  \frac{\sum_t \exp(-\beta\sum_{t^{\prime}\leq t} a_{t'} \ln\frac{ t}{t^{\prime}})}{\sum_t ta_t} \qquad \text{ s.t. }\  \|a\|_1 =1, \qquad a_t \geq 0\,.  \]

\begin{claim}
\label{cl:optmssc}
For any $\beta > 1$, the optimum value of the problem $\P$ is at most $\beta/(\beta-1)$.
\end{claim}
\begin{proof}
Consider the numerator of the objective in $\P$. Each summand $\exp(-\beta \sum_{t' \leq t} a_{t'} \ln (t/t'))$ is of the form  $f=\exp(-\sum_{t'} c_{t'} a_{t'})$  
(note $c_{t'} = \beta \ln (t/t')$ are constants), and hence is convex in $a$. As the sum of convex functions is convex, the numerator is convex. 
Moreover, the domain $H$ of $a$  is the non-negative simplex and the denominator of $\P$ is linear and non-negative on $H$. By Fact \ref{fact}, the optimum of $\P$ is at some extreme point of the unit simplex, which must be of the form $a_u = 1$ for some $u$ and $a_{u'}=0,u' \neq u$.

For such an extreme point, the objective of $\P$ is easy to compute. In particular, the denominator is exactly $u$ and the numerator is 
\[  \sum_{t=1}^u  1  +  \sum_{t > u}  \exp (-  \beta \ln t/u) = u  + \sum_{t > u} (u/t)^\beta \leq u + \int^\infty_{x=u} (u/x)^\beta dx = u\Bigl(1 + \frac{1}{\beta-1}\Bigr) = \frac{\beta}{\beta-1} u.\qedhere \]
\end{proof}
\noindent This proves that $c_z(e)\leq \beta/(\beta-1) c_x(e)$ as desired. 
 \end{proof}
 
\subsection{Min Latency Set Cover}
\label{sec:minlat}
We now consider the min latency setting where $k_e  = |e|$, and show the following.
\begin{theorem}
\label{thm:lat} For any $\beta \geq 1$, the algorithm is a
$\beta e^{1/\beta}$-approximation. In particular for $\beta=1$, this gives an $e$-approximation.\footnote{The use of $e$ both for an edge and the constant $e=2.718\ldots$ is a bit unfortunate, but hopefully it is not confusing.}
\end{theorem}

For min latency, the KC constraints by \eqref{eq:lpcov} imply, for each $e,t$ and $v \in e$,  that $u_{e,t} + x_{v,<t} \geq 1$. In particular, we can assume $u_{e,t} = \max_{v\in e}(1-x_{v,<t})$.
Let us define $x_{e,t} = u_{e,t}-u_{e,t+1}$ (the amount of $e$ satisfied by the LP at time $t$), and note that $x_{e,<t} = 1-u_{e,t}$. So we have that $x_{v,<t} \geq x_{e,<t}$ for all $v \in e$ and all $t$, i.e., $x_v \succeq x_e$.

For an edge $e$, let $t_e$ be the earliest time  such that
$z_{v, \leq t_e} \geq 1$ for all $v \in e$. We define $c_z(e)=t_e$. Note that the $\alpha$-point rounding will surely schedule each $v \in e$ in $\tau$ by time  $t_e$, and hence $\cov_\tau(e) \leq c_z(e)$ with probability $1$. 

Theorem~\ref{thm:lat} is a direct consequence of the following Lemmas~\ref{lem:ketent} and \ref{cl:ketentofin}.

\begin{lemma}\label{lem:ketent} For each edge $e$, \ $c_z(e) \leq e^{1/\beta} c_x(e)$.
\end{lemma}
\begin{proof}
Fix an edge $e$. Let $t^*$ (not necessarily an integer) be such that
\beq \label{eq:t*-lat} \beta\sum_{t' \leq \lfloor t^* \rfloor} x_{e,t'}\ln\frac{t^*}{t'} = 1\,. \eeq
We claim that $t_e \leq \lfloor t^* \rfloor$. This follows as for any $v \in e$, by Claim \ref{cl:kernelog} with $t= \lfloor t^* \rfloor+1$, 
\[ z_{v,\leq \lfloor t^* \rfloor} = (K x_v)_{\leq \lfloor t^* \rfloor} \geq (K x_e)_{\leq \lfloor t^* \rfloor} \geq \beta \sum_{t' \leq \lfloor t^* \rfloor} x_{e,t'} \ln \frac{\lfloor t^* \rfloor+1}{t'} \geq \beta \sum_{t' \leq \lfloor t^* \rfloor} x_{e,t'} \ln \frac{t^*}{t'} = 1\]
where  the first and second inequalities use Claim~\ref{cl:kernelog} and $x_{v} \succeq x_{e}$, and the last equality uses \eqref{eq:t*-lat}.

As $c_z(e) = t_e \leq \lfloor t^* \rfloor \leq t^*$, to prove the lemma it suffices to show that $c_x(e) \geq e^{-1/\beta} t^*$. As $c_x(e)=\sum_t t x_{e,t}$, this is equivalent to showing that the optimum solution of the following optimization problem with variables $a_t$ (corresponding to $x_{e,t}$), is at least $ e^{-1/\beta} t^* $. 
\[ (\P) \label{def:probke} \qquad 
 \text{Minimize}\  \sum_t ta_t \qquad \text{ s.t. } \|a\|_1 =1 \ \mbox{ and } \ \ \  \beta \sum_{t \leq \lfloor t^*\rfloor} a_t \ln \frac{t^*}{t} = 1\,.  \]
The first constraint follows as $x_{e,\leq T} =1$ for $T$ large enough (since there are $n$ vertices, we can assume that $x_{v,\leq n} = 1$ for all $v$), and the second constraint is by the relation \eqref{eq:t*-lat} defining  $t^*$.
\begin{claim} The optimum of $\P$ is at least $t^* e^{-1/\beta}.$ \end{claim}\label{cl:optk1}
\begin{proof}
As $\P$ is an LP with 
$2$ non-trivial constraints, there is some optimum solution with at most $2$ non-zero $a_{t}$. Let $u$ and $v$ be those indices with $u \leq v$ and let $a_u = s$ and $a_v = 1-s$ so that $\|a\|_1=1$. 

We consider two cases, depending on whether $v \leq \lfloor t^* \rfloor$ or not.  Suppose first that $v \leq \lfloor t^* \rfloor$.
    Then $(\P)$ reduces to the following problem on three variables $u,v,s$.
\[ \text{Minimize } su + (1-s)v \qquad \text{ s.t. } \qquad s \ln (t^*/u) + (1-s) \ln (t^*/v)  = 1/\beta, \qquad s \in [0,1], u \leq v \leq \lfloor t^*\rfloor \,.\]
Setting $u=t^*e^{-c}$ and $v = t^*e^{-d}$ where $c,d\geq 0$ as $u,v\leq t^*$, this becomes
\[ \text{Minimize } t^*(se^{-c} + (1-s)e^{-d}) \qquad \text{ s.t. } \qquad (sc + (1-s)d)  = 1/\beta, \qquad s \in [0,1], c\geq d\geq 0.\]
But as $e^{-x}$ is convex, the objective is always at least $t^*(e^{-(sc + (1-s)d)}) = t^*(e^{-1/\beta})$, as claimed.

\noindent Next, suppose that $v > \lfloor t^* \rfloor$. The value $a_v$ does not affect $t^*$ here, and $(\P)$  now becomes
\[ \text{Minimize } su + (1-s)v \qquad \text{ s.t. } \qquad s \ln (t^*/u)  = 1/\beta, \qquad s \in [0,1], u \leq \lfloor t^* \rfloor , v > \lfloor t^* \rfloor \,. \]
The constraint implies $u = t^* e^{-1/s\beta}$, and let us further relax the constraints to $u \leq t^* , v \geq  t^*$ as this can only reduce the objective. 
As the only constraint of $v$ is $v \geq t^*$, $v=t^*$ in any optimum solution.

So the problem reduces  to minimizing $ t^* (s e^{-1/s\beta} + (1-s)) $ for  $s \in [0,1] $. 
Setting $f(s) = (s e^{-1/s\beta} + (1-s))$ and checking that $f'(s) = e^{-1/s\beta}(1+1/s\beta) - 1 \leq 0$ for $s \geq 0$ (as $1+x \leq e^x$ for all $x$), we get that it is minimized at $s=1$ with value $t^* e^{-1/\beta}$.
This completes the proof of the claim.
\end{proof}
\noindent Putting all this together, $c_z(e) \leq t^*$ and $c_x(e) \geq t^* e^{-1/\beta} $, which completes the proof of Lemma~\ref{lem:ketent}. 
\end{proof}

Finally we relate the expected cost under $\sigma$ to $c_z(e)$. The argument is in fact quite simple for min-latency.
\begin{lemma}\label{cl:ketentofin} For each edge $e$, $\expec{}{\cov_\sigma(e)} \leq \beta t_e = \beta c_z(e)$.
\end{lemma}
\begin{proof}
As all the vertices of $e$ are scheduled by $t_e$ in $\tau$, the cover time $\cov_\sigma(e)$ of $e$ in $\sigma$ is trivially upper bounded by the number of vertices not in $e$ that appear in $\tau$ by time $t_e$, plus $k_e$,
the number of vertices in $e$ needed to cover $e$.

As $\sum_{v} z_{v,\leq t_e} \leq \beta t_e$ by Claim \ref{cl:widthz}, and as $z_{v,\leq t_e} \geq 1$ for each $v \in e$,
the expected number of vertices not in $e$ that appear by time $t_e$ is $\sum_{v \notin e} z_{v,\leq t_e} \leq \beta t_e - k_e$. So
 $\expec{}{\cov_\sigma(e)} \leq \beta t_e -k_e + k_e = \beta t_e$. 
\end{proof}

\subsection{Generalized Min Sum Set Cover}
\label{sec:gmssc}
We now consider the setting with general $k_e$. The algorithm is exactly the same as before: we use the kernel $K(t,t')=\beta/t$ for $t'\leq t$, and follow the algorithmic template.
Our analysis will (necessarily) combine aspects from both MSSC and the min-latency setting, and will be more technically involved. We will also need to use the KC inequalities in a careful way.

Briefly recalling earlier notation, let $x$ denote some optimum LP solution, and let $z$ be obtained by applying the kernel $K$ to $x$, and let $\tau$ be the tentative schedule obtained by applying the $\alpha$-point rounding to $z$ for each $v$.

Fix a time $t$.
For each vertex $v$, let $Y_v$ be a Bernoulli random variable which is $1$ if $v$ is picked before $t$ in schedule $\tau$ and $0$ otherwise.
Then, by the property of $\alpha$-point rounding
$\expec{}{Y_v}=\min(1,z_{v,<t})$, the random variables $Y_v$ for different vertices $v$ are independent, and the probability that an edge $e$ is not yet covered before $t$ in $\tau$ can be written as
\begin{align}
    \label{def:p_t}
    p_t(e) = \Pr[\sum_{v \in e} Y_v < k_e]\,.
\end{align}

\noindent\textbf{ The $P(\cdot)$-function.} 
As $\expec{}{\cov_\tau(e)} = \sum_t p_t(e)$,
this motivates understanding the following question:

Given Bernoulli random variables ${Y_v}$ with $\expec{}{\sum_v Y_v} = \gamma k$, where $k$ is a positive integer and $\gamma \geq 1$, what is the best
upper bound $P(\gamma)$ on $\Pr[\sum_v Y_v \leq k-1]$,
as a function of $\gamma$?

\noindent
{\em Remark.} Note that we require $P(\gamma)$ to only depend on $\gamma$ (as $k$ and $\expec{}{Y_v}$, will be completely arbitrary for us). Moreover, we assume that $\gamma\geq 1$, since if $\gamma<1$,  we simply assume the worst case that $p_t(e)=1$.

Let us first note that 
for $k=1$, this question is precisely $\Pr[\sum_v Y_v=0]$ if $\expec{}{\sum_v Y_v}=\gamma$, which tends to $\exp(-\gamma)$ (e.g.~in the Poisson regime). So the best $P(\gamma)$ we can hope for is $\exp(-\gamma)$. However, in general $P(\gamma)$ must be strictly worse: in particular, if $k$ is large and $\gamma=1$, then $\Pr[\sum_v Y_v \leq k-1]$ is arbitrarily close to $1/2$, see \cite{JS68}.

One of our main technical contributions is coming up with refined functions $P(\gamma)$ that are essentially optimum\footnote{One could use standard Chernoff-Hoeffding bounds or other tail bounds for sums of Bernoulli random variables, but these are quite crude for our purposes as we lose relatively large constants, leading to substantially worse approximation factors.}. In particular Theorem \ref{thm:enhanced-bound} provides us with
\beq \label{bd:stronger-eq} P(\gamma) \leq e^{-\gamma} (e/2)^{\exp(1-\gamma)}\,.\eeq
Note that here, $P(1)=1/2$ (the best we can hope for). Moreover, and as $\gamma$ increases, the factor  $(e/2)^{\exp(1-\gamma)}$ approaches rapidly to $1$ (as $\exp(1-\gamma)$ approaches 0), and hence  $P(\gamma)$ is essentially $e^{-\gamma}$ for larger $\gamma$. So the bound is close to the best possible for the entire range of $\gamma$.
The derivation of this bound is rather technical and is deferred to the Appendix. In \S\ref{sec:tail}, we show a weaker (but still quite non-trivial) bound of $P(\gamma) = e^{1-\gamma}/2$.

\noindent\textbf{Our result.}
Our main result is the following, which we state in terms of a general function $P(\cdot)$.
\begin{theorem}\label{thm:maingmssc}
For any $\beta \geq 1$, and any tail bound function $P(\cdot)$ that is convex and non-increasing on $[1,\infty)$ and satisfies $P(1)\leq 1$, the approximation ratio of the algorithm is at most $ \beta r(\beta)$ where 
\begin{equation}
    \label{gen-bound}
r(\beta) := e^{1/\beta}\left(1+ \frac{1}{\beta} \int^\infty_1 P(x) e^{\frac{x-1}{\beta}} dx \right)\,.
\end{equation}
In particular, for the bound $P(\gamma) = e^{1-\gamma}/2$, this gives the approximation
$\beta (2\beta-1) e^{1/\beta}/(2\beta-2)$
which has value at most $4.9102$, for $\beta=2.191$. Using the  stronger bound \eqref{bd:stronger-eq} for $P(\gamma)$, the approximation factor is at most $4.642$,
 for $\beta  = 2.0715$.\end{theorem}

\noindent{\em Remark.}
Based on numerical computations, it seems that $P(\gamma)$  satisfies the stronger bound   
\[ \label{bd:stronger} P(\gamma) \leq e^{-\gamma} (e/2)^{\exp(c(1-\gamma))}\,,\]
for $c=13.3$ (note that the bound in \eqref{bd:stronger-eq} corresponds to $c=1$). However, we are unable to prove this analytically.
Plugging this better bound for $P$ gives an approximation ratio of at most $4.5232$.

We do not pursue  further optimizations as
$4.509$ seems a natural bottleneck for our approach, as even assuming the best possible function $P(\gamma) = \exp(-\gamma)$, Theorem \ref{thm:maingmssc} only gives a $4.509$ approximation
(and hence additional ideas will be needed to get a $4$-approximation, assuming this is possible).

\noindent\textbf{Analysis.} Fix an edge $e$.
and as previously, define $x_{e,t}=u_{e,t}-u_{e,t+1}$ (the fraction of $e$ satisfied by the LP at $t$). Then $c_x(e) = \sum_t u_{e,t} = \sum_t tx_{e,t}$. 
Let $z_e = K x_e$ be obtained by applying the kernel $K$ to $x_e$. Let $t_e$ be the earliest time such that $z_{e,\leq t_e} \geq 1$. Let us define
\beq \label{czegmssc} c_z(e) = t_e + \sum_{t>t_e} P(z_{e,<t}).\eeq
We will show that $c_z(e)$ is an upper bound on the expected cost $\expec{}{\cov_\tau(e)}$ of $e$ in the tentative schedule $\tau$.
Theorem \ref{thm:maingmssc} will follow directly from the following two lemmas.
\begin{lemma}
\label{lem:tentcostgmssc}
For any edge $e$, $c_z(e) \leq r(\beta)  c_x(e)$.
\end{lemma}
\begin{lemma}\label{lem:tentofingmssc}
For each edge $e$, $\expec{}{\cov_\sigma(e)} \leq \beta c_z(e)$.
\end{lemma}

\begin{proof}(Lemma \ref{lem:tentcostgmssc})
We first show that $\expec{}{\cov_\tau(e)} \leq c_z(e)$.

Recall by \eqref{def:p_t} that $p_t(e)$ denotes the probability that $e$ is not covered before $t$ in $\tau$.
Then $\expec{}{\cov_\tau(e)}= \sum_t p_t(e)$. To upper bound $p_t(e)$ we assume that $p_t(e)=1$ for $t\leq t_e$.  
Now, fix some time $t > t_e$. We claim that $p_t(e) \leq P(z_{e,<t})$. This directly gives that $\expec{}{\cov_\tau(e)} \leq c_z(e)$,  by the definition of $c_z(e)$ in \eqref{czegmssc}.

Let us define $A_e(t):= \{v\in e: z_{v,< t} \geq 1 \}$ to be the set of vertices in $e$ that are surely picked in schedule $\tau$ before time $t$, and let $B_e(t) = e \setminus A_e(t)$ be the remaining vertices in $e$. Let $k_e(t) = \max(k_e - |A_e(t)|,0)$ be the number of vertices that must still be picked from $B_e(t)$ to satisfy $e$.

By the KC inequality for $e$ with $S=A_e(t)$, the LP satisfies 
\[ u_{e,t} + \sum_{v \in B_e(t)} x_{v,<t}/(k_e(t)) \geq 1\,, \]
or equivalently, $\sum_{v \in B_e(t)} x_{v,<t} \geq k_e(t)
x_{e,<t}$.
Applying the kernel K on both sides then gives (via Claim~\ref{cl:kernelog})
 \beq 
 \label{kc:consequence}
 \sum_{v \in B_e(t)} z_{v,<t} \geq k_e(t) z_{e,<t}\,.\eeq  
Now $e$ is not covered in $\tau$ before $t$, if and only if strictly fewer than $k_e(t)$ vertices from $B_e(t)$ are picked in  $\tau$. Using the random variables $Y_v$ described earlier, this is exactly if $\sum_{v \in B_e(t)} Y_v < k_e(t)$, and hence from the definition of $P(\cdot)$ and \eqref{kc:consequence}, it follows that  $p_t(e) \leq P(z_{e,<t})$, as claimed.

We now upper bound $c_z(e)$ in terms of $c_x(e)$.
Let $t^*$ (not necessarily integer) be such that
\[ \beta\sum_{t' \leq \lfloor t^* \rfloor} x_{e,t'} \ln (t^*/t') = 1. \]
By Claim \ref{cl:kernelog}, as $z_e= Kx_e$, we have $z_{e,\leq t} \geq \beta \sum_{t' \leq t} x_{e,t'} \ln (t+1)/t'$. Hence, $z_{e,\leq \lfloor t^* \rfloor} \geq 1$ implying that $\lfloor t^* \rfloor \geq t_e$. 
Moreover as $P(\gamma)$ is non-increasing for $\gamma \geq 1$ (by our hypothesis), we have 
\[ c_z(e) = t_e + \sum_{t > t_e} P(z_{e,< t}) \leq  \lfloor t^* \rfloor  + \sum_{t > \lfloor t^* \rfloor} P(z_{e,<t}) \leq \lfloor t^* \rfloor  + \sum_{t > \lfloor t^* \rfloor} P\Bigl(\beta \sum_{t^{\prime}\leq t} x_{e,t^{\prime}} \ln{\frac{t}{t^{\prime}}}\Bigr)\,.  \]

Our goal henceforth is to upper bound the  maximum possible ratio of the right hand side above to $c_x(e)$.

\noindent {\bf Optimization problem.} As previously, we set up the following optimization problem $\P$ with variables $a_t$ (corresponding to $x_{e,t}$), and some fixed real $t^*$. 
\[ (\P) \quad \text{Maximize }  \frac{\lfloor t^* \rfloor   + \sum_{t > \lfloor t^*\rfloor} P(\beta(\sum_{t^{\prime}\leq t} a_{t'} \ln t/t^{\prime}))}{\sum_t ta_t} \qquad \text{ s.t. } \|a\|_1 =1, \qquad \beta \sum_{t^{\prime}\leq \lfloor t^* \rfloor} a_{t'} \ln (t^*/t^{\prime}) = 1\,.  \]
We bound this in Lemma \ref{cl:psolgmssc} which will imply the result.
 \end{proof}

\begin{lemma}\label{cl:psolgmssc}
If $P$ is non-increasing and convex in $[1,\infty)$, and satisfies $P(1) \leq 1$. Then, the optimal value of problem $\P$ is at most $r(\beta)$, defined in \eqref{gen-bound}.
\end{lemma}
\begin{proof}
Let us denote $\gamma_t = \beta \sum_{t^{\prime}\leq t} a_t \ln t/t^{\prime}$.  The constraint in $\P$ implies that, $\gamma_t \geq 1$ for $t \geq t^*$. Also, as $\gamma_{t^*}=1$, we have $P(\gamma_{t^*}) \leq 1$ by the assumption on $P$.

As the constraints of $\P$ are linear in $a$, the feasible region is a polytope. Moreover as $\gamma_t$ is linear in $a$ and $P$ is convex, $P(\gamma_t)$ is convex in $a$. So the objective of $\P$ is the ratio of a non-negative convex function and a strictly positive linear function, and by Fact~\ref{fact}, attains its optimum at an extreme point.
As there are only $2$ non-trivial constraints, any extreme point has at most $2$ non-zero variables, say  $a_u$ and $a_v$ for $u\leq v$. 

We consider two cases depending on whether $v\leq t^*$.

\noindent {\em Case 1:} $v\leq t^*$.  Let us denote $a_u = s$, then $a_v = 1-s$ (as $\|a\|_1=1$). 
Let us substitute $u = t^* e^{-a}$ and $v=t^*e^{-b}$, where $a,b \geq 0$. Then the second constraint in $\P$ is
$\beta(sa + (1-s)b)=1$. Moreover,
\[  \gamma_t = \beta \left(s \ln (t/u) + (1-s) \ln(t/v)\right) = \beta \ln t/t^* + \beta(sa+(1-s)b) = 1 + \beta \ln (t/t^*)\,. \]
So, the numerator of the objective of $\P$ is simply $ \lfloor t^* \rfloor + \sum_{t > \lfloor t^* \rfloor} P\left(1 + \beta \ln (t/t^*)\right)$,
which surprisingly is completely independent of the parameters $a_u,a_v,s$.

So optimizing $\P$ simply reduces to 
minimizing the denominator. In particular,
\[ \text{Minimize }\, t^* (s e^{-a} + (1-s) e^{-b})  \qquad \text{s.t. } sa  + (1-s) b =1/\beta, \quad  a \geq b \geq 0, s \in [0,1].\]
As $\exp(-x)$ is convex, by Jensen's inequality $se^{-a} + (1-s) e^{-b} \geq e^{-s a  - (1-s) b} = e^{-1/\beta}$, and hence the denominator is lower bounded by $t^* e^{-1/\beta}$.

So any optimum solution of $\P$ is at most  
\begin{align*}
e^{1/\beta} \Big( \frac{\lfloor t^* \rfloor}{t^*} + \frac{1}{t^*} \sum_{t > \lfloor t^* \rfloor} P\left(1 + \beta \ln (t/t^*)\right)\Big) &\leq e^{1/\beta} \left( 1 + \frac{1}{t^*} \int^\infty_{t=t^*} P\left(1 + \beta \ln (t/t^*)\right)dt\right)\\
&=e^{1/\beta}\left(1 + \frac{1}{\beta}\int^\infty_{x=1}P(x) e^{\frac{x-1}{\beta}} dx\right) = r(\beta)\,,
\end{align*}
where the first inequality follows as $P$ is non-increasing and as $P(\gamma_{t^*}) \leq 1$. The second equality follows by the change of variables $x=1+\beta \ln(t/t^*)$.

\noindent {\em Case 2:} $v > t^*$. 
Again, let $a_u = s$ and $a_v = 1-s$. The second constraint of $\P$ now becomes $\beta s \ln(t^*/u) = 1$ and hence $u = t^* =t^* \exp(-1/s\beta)$.
Moreover, the numerator of the objective of $\P$ is
\[ \lfloor t^*\rfloor + \sum^v_{t=\lfloor t^*\rfloor+1} P\left(\beta s \ln (t/u)\right) + \sum^\infty_{t=v+1} P\left(s\beta \ln (t/u) +(1-s)\beta \ln (t/v)\right)\,.\]
Substituting the value of $u$ gives $\beta s\ln (t/u) = 1 + \beta s \ln(t/t^*)$. Further, setting $v = t^*a$, the numerator simplifies to
\[  \lfloor t^*\rfloor + \sum^v_{t=\lfloor t^*\rfloor+1} P\left(1+\beta s \ln (t/t^*)\right) + \sum^\infty_{t=v+1} P\left(1+\beta \ln (t/t^*) - (1-s)\beta \ln a \right)\,.\]
Substituting the variable $x = t/t^*$, upper bounding sums by integrals using that $P$ is non-increasing and that $P(\gamma_{t^*})\leq 1$, we finally get the following upper bound on the numerator
\beq\label{eq:numvgt}
 t^*\left( 1 + \int^a_1 P(1+s \beta \ln x)dx + \int^\infty_a P(1+\beta \ln x - \beta (1-s) \ln a\right)\,.
 \eeq

Again, by the substituting the value of $u$ and $v=t^*a$,  the denominator now becomes
\begin{equation}\label{eq:denvgt}
su + (1-s)v = t^*(s e^{-1/(s\beta)} + a(1-s)).
\end{equation}

Fix $s,\beta$, and let us define $f(a)$ to be expression  in (\ref{eq:numvgt}) divided by $t^*$ and $g(a)$ to be expression in (\ref{eq:denvgt}) divided by $t^*$. Our goal is to show that $\max_{a \geq 1} f(a)/g(a) \leq r(\beta)$.

The next key claim, the proof of which is in Section \ref{sec:maxa1} in the Appendix, shows that it suffices to consider the value of $f$ and $g$ at $a=1$.
\begin{claim}\label{cl:maxa1}
For any $s \in [0,1)$ and $\beta \geq 1$, the ratio $f(a)/g(a)$, for $a \geq 1$, is maximized at $a=1$.
\end{claim}
Now for $a=1$, we have $g(1) = se^{-1/(s\beta)} + 1 - s$ which by simple calculus is minimized at $s=1$ (for $s \in [0,1]$) and hence is at least  $\geq e^{-1/\beta}$. Next,
\[ f(1) = 1 + \int^\infty_1 P(1+\beta \ln x) dx = 1 + \frac{1}{\beta} \int^\infty_1 P(y) e^{\frac{y-1}{\beta}} dy  = r(\beta) e^{-1/\beta}\,.\]
So, $f(a)/g(a) \leq r(\beta)$, which gives the required bound for Lemma~\ref{cl:psolgmssc}.
\end{proof}

\noindent\textbf{Relating the final schedule $\sigma$ to $c_z(e)$.}
We now prove Lemma~\ref{lem:tentofingmssc}, that $\expec{}{\cov_\sigma(e)} \leq \beta c_z(e)$.

\begin{proof}(Lemma \ref{lem:tentofingmssc})
Fix a time $t$, and condition on the event $E_t$ that $e$ is covered at time $t$ in $\tau$, that is, $\cov_\tau(e) = t$.
Recall that for all $t$, $\Pr[\cov_\tau(e) \geq t] \leq P(z_{e < t})$  (with $P(\gamma)=1$ for $\gamma <1$).
So to upper bound $\expec{}{\cov_\sigma(e)}$, we can assume that $\Pr[E_t] = P(z_{e < t}) -  P(z_{e \leq t})$, and that $t \geq t_e$, as otherwise $P(z_{e \leq t})= 1$ and hence $\Pr[E_t]=0$.

Now conditioned on the event $E_t$, the cover time $\cov_\sigma(e)$ of $e$ is at most $k_e + \expec{}{N(t)}$ where $N(t)$
 are the vertices not in $e$ that appear at time $\leq t$ in $\tau$, and the delay of $k_e$ is due to the vertices of $e$ itself. As the rounding for vertices $u \notin e$ is independent of $E_t$, we have that
 \[\expec{}{N(t)} = \sum_{u \notin e} z_{u,\leq t} \leq 
\beta t - \sum_{v \in e} z_{v,\leq t}\,.\]
Let $A_e(t)$ be the set of vertices $v$ in $e$ with $z_{v,\leq t} \geq 1$, and $B_e(t) = e \setminus A_e(t)$. Then,
\[\sum_{v \in e } z_{v,\leq t} = \sum_{v \in A_e(t)} z_{v,\leq t} + \sum_{v \in B_e(t)} z_{v,\leq t} \geq |A_e(t)| + k_e-|A_e(t)| = k_e\,,\]
where we use \eqref{kc:consequence}  and that $z_{e,\leq t} \geq 1$ for  $t \geq t_e$. This gives that,
\[
\expec{}{\cov_\sigma(e)} \leq \sum_t \Pr[E_t] ( k_e + \expec{}{N(t)}) \leq  \Pr[E_t] (k_e + (\beta t - k_e)) =  \beta \sum_t \Pr[E_t]\, t.\]
As $\sum_t \Pr[E_t]\, t = \sum_t  t  (P(z_{e < t}) -  P(z_{e \leq t})) = \sum_t P(z_{e < t}) = c_z(e)$, the result follows.
 \end{proof}

\section{Min Sum Vertex Cover}\label{sec:MSVC}
Now we consider MSVC, a special case of GMSSC when $|e|=2$ and $k(e)=1$. We will apply the same framework as before, but with a different kernel $K$ and use it to obtain a $16/9=1.777\ldots$ approximation.

We use the same LP as that for MSSC, where now  constraint \eqref{eq:lpcovbasic} can be written more explicitly as 
\[ u_{e,t} \geq 1- (x_{v,<t} + x_{w,<t}) \qquad \qquad \forall e=(v,w).\]
We will also show a matching $16/9$ integrality gap for this LP.

 \subsection{Approximation Result} We apply our generic Algorithm \ref{meta-algorithm} with the following kernel, 
\[K(t,t') = \frac{4t^{\prime}(t^{\prime}+1)}{t(t+1)(t+2)} \cdot \mathbbm{1}[t \ge t'].\]
and show the following
\begin{theorem}\label{thm:MSVC-approx}
The algorithm, with the above kernel, is a $16/9$-approximation for MSVC.
\end{theorem}

This kernel should be viewed intuitively as $4t'^2/t^3$, but as time is discrete in our formulations, we use the version above.
This kernel is qualitatively very different from the one we used previously for MSSC. The idea is that in the vertex cover problem, for any edge $e=(v,w)$ at least one of $v,w$ (say $v$) will eventually be picked to an extent of at least $1/2$.
So we only need a kernel that ensures that $v$ has a mass of $1$ in the solution $z$ as $x_{v,\leq t}$ approaches $1/2$ for $t$ large enough. This allows us to pick a kernel where $K(t,t')$ decays more rapidly with $t$ (in particular, as $1/t^3$ here as opposed to $1/t$ in MSSC).

Let us make the above idea precise. 
We first state a couple of standard identities that we will need.
\[\sum_{t=a}^b t(t+1) = \frac{1}{3} \Bigl( b(b+1)(b+2) - a(a-1)(a+1)\Bigr)\,.\]
Similarly, writing $1/t(t+1) = 1/t - 1/(t+1)$ and $2/t(t+1)(t+2) = 1/t - 2/(t+1) - 1/(t+2)$,
\[ \sum_{t=a}^b  \frac{2}{t(t+1)(t+2)} =  \frac{1}{a(a+1)} - \frac{1}{(b+1)(b+2)}\,.\]

\begin{lemma} $\sum_v z_{v,t} \leq 4/3$\,, for any $t$.
\end{lemma}
\begin{proof}
By  constraint \eqref{eq:lppack}, \ $\sum_v x_v \leq \bf{1}$, and hence 
$
\sum_v z_{v}
= \sum_v Kx_v \leq K \mathbf{1}$. Now,
\[ (K\mathbf{1})_t = \sum_{t'} K(t,t') = \sum_{t^{\prime} \leq t} \frac{4t^{\prime}(t^{\prime}+1)}{t(t+1)(t+2)} = 
\frac{4}{t(t+1)(t+2)} \cdot \frac{t(t+1)(t+2)}{3} = 4/3.
\qedhere \]
\end{proof}

As usual, to prove Theorem \ref{thm:MSVC-approx}, it suffices to show the following.
\begin{lemma}\label{lem:4/3}
For any edge $e$,  $\expec{}{\cov_\tau(e)} \leq (4/3) c_x(e)$.
\end{lemma}
\begin{lemma}\label{lem:cond-msvc}
For any edge $e$,  $\expec{}{\cov_\sigma(e)} \leq (4/3) \expec{}{\cov_\tau(e)}$.
\end{lemma}
The proof of Lemma~\ref{lem:cond-msvc} is given in Section~\ref{sec:conditioning}, and here we prove Lemma~\ref{lem:4/3}.
\begin{proof}[Proof of Lemma~\ref{lem:4/3}]
Fix an edge $e = (v,w)$.
Let $p_{v,t}$ denote the probability that $v$ is not covered before time $t$ in the tentative schedule $\tau$. Then, the expected cover time of $e$ in $\tau$ is $\expec{}{\cov_e(\tau)}  = \sum_{t \ge 1} p_{v,t} \cdot p_{w,t}$. To compute $p_{v,t}$, note that
the $v$ is assigned to a slot $<t$ by $\alpha$-point rounding with probability $\min \set{1, z_{v,<t}}$, and
\[ z_{v, < t} = \sum_{t'<t} z_{v,t'} 
= \sum_{t' < t} \sum_{t'' \leq t'} K(t',t'') x_{v,t''}
=\sum_{t'' < t} x_{v,t''} \sum_{t' = t''}^{t-1}
K(t',t'') = \sum_{t'' < t} 2\left(1-\frac{t''(t''+1)}{t(t+1)}\right) x_{v,t''} \]
where we use that 
\[\sum_{q = t''}^{t-1} K(q,t'')= {4 t'' (t''+1)} \sum_{q = t''}^{t-1} \frac{1}{q(q+1)(q+2)} = 2 \left(1 - \frac{t''(t''+1)}{t(t+1)} \right).\] 
This gives that,
\beq 
\label{pvt} p_{v,t}= 1 -\min \set{1, z_{v,<t}}  =  \left( 1- \sum_{ t^{\prime}< t } 2\left(1-\frac{t'(t'+1)}{t(t+1)}\right) x_{v,t'} \right)_+\,. \eeq
Next, the LP cost $c_x(e)$ for $e$ is exactly
$\sum_t \left( 1- \sum_{t' < t} (x_{v,t'} + x_{w,t'}) \right)_+$. 
So our goal is to show that for any setting of $x_{v,t}$ and $x_{w,t}$, 
\beq
\label{eq:toshowmsvc} \sum_t p_{v,t} \cdot p_{w,t} \leq \frac{4}{3} \sum_t \left( 1- \sum_{t' < t} (x_{v,t'} + x_{w,t'}) \right)_+ \eeq
where $p_{v,t}$ and $p_{w,t}$ are given by \eqref{pvt}.

To this end, we first make some simple observations. First, we can assume that $x_{v,<t} + x_{w,<t} \leq 1$ for all $t$ as this does not affect the right side of \eqref{eq:toshowmsvc}, and can only increase the left hand side.
Next, replacing each $x_{v,t}$ and $x_{w,t}$ by their average value $(x_{v,t}+x_{w,t})/2$ does not affect the right side, and can only increase the left side as by AM-GM inequality, $ (1-a)(1-b) \leq (1-(a+b)/2)^2$\,,
whenever $a,b \in [0,1]$.
Thus proving \eqref{eq:toshowmsvc} reduces to showing the following. 

Given a vector $a$ with non-negative entries (where $a_t$ corresponds to 
$x_{v,t}+x_{w,t}$) and $\|a\|_1=1$, 
\[
\sum_t \left( 1- \sum_{ t^{\prime}< t} \left(1-\frac{t'(t'+1)}{t(t+1)}\right) a_{t'} \right)^2 \leq \frac{4}{3}  \sum_t \left( 1- \sum_{t' < t} a_{t'} \right)\,. 
\]
Let $f(a)$ denote the left hand side above, and $g(a)$ denote the right side, then we get the following problem.
\[ (\P) \qquad  \label{def:probmsvc} \text{ Maximize} \qquad   f(a)/g(a) \qquad \text{s.t.} \qquad \|a\|_1 =1, \qquad a_t \geq 0. \]
Again, $g(a)$ is linear in $a$ and  strictly positive over the positive simplex. For a fixed $t$, each summand in $f(a)$ is a term of the form $(1- c \cdot a)^2$\,, where $c$ is a vector with $\|c\|_\infty \leq 1$, implying that $|c \cdot x|\leq 1 $ and that the term is convex. Thus $f(a)$ is convex. So by Fact \ref{fact}, the optimum is attained at some extreme point of the simplex of the form $a_u=1$ and $a_{u'}=0$ for $u'\neq u$.

In this case, $g(a) = {4u/3 }$ and 
\[ f(a) = u + \sum_{t > u}  \Big(\frac{u(u+1)}{t(t+1)}\Big)^2  \leq u + \sum_{t >u}  \frac{u^2(u+1)^2}{(t-1)t(t+1)(t+2)} =   u + \frac{u^2(u+1)^2}{3u(u+1)(u+2)} \leq {\frac{4u}{3}}, \]
which proves the desired inequality.
\end{proof}

\subsection{Tight Integrality Gap}
We now show that the integrality gap of the LP for MSVC is arbitrarily close to $16/9$.

Consider an instance consisting of disjoint copies of the complete graphs $K_i$, for $i=1,\ldots,k$, where $K_i$ has $n_i$ vertices. We will set $n_i = N i^{-\alpha}$ where $\alpha=\frac{2}{3}+\epsilon$, with $\epsilon$ approaching $0$, and for a suitable $k=\omega(1)$. 
Also choose $N$ large enough 
(so the issues of floor and ceiling in $n_i$ will not affect us). Note that $n_1> n_2 \ldots > n_k$.
To show the integrality gap, we will upper bound the LP cost suitably, and show that any integral solution has cost at least $16/9 -O_{\epsilon}(1)$ times the upper bound on the LP cost.

\begin{lemma} \label{lem:iglpcost}The LP cost is at most $N^3/{(4\epsilon)}$, up to lower order terms.
\end{lemma}
\begin{proof}
Consider the fractional solution that first schedules each vertex of $K_1$ to extent $1/n_1$ in the first $n_1/2$ slots, then schedules the vertices of $K_2$ in the next $n_2/2$ slots and so on.
Then, each edge in $K_i$ is completely covered by time $(n_1+\ldots+n_i)/2$, and the cost of this solution is at most
\[ \sum_{i=1}^k {n_i \choose 2} \frac{n_1 + \ldots n_i}{2} \leq \sum_{i=1}^k \frac{N^3}{4} i^{-2\alpha} \sum_{j=1}^i j^{-\alpha} \leq \sum_{i=1}^k \frac{N^3}{4} i^{-2\alpha} \frac{ i^{1-\alpha}}{1-\alpha} = \frac{N^3}{4(1-\alpha)} \sum_{i=1}^k{ i^{1-3\alpha}}\,, \]
where we used that $\sum_{j=1}^i j^{-\alpha} \leq \int_{j=0}^i x^{-\alpha} dx  = i^{1-\alpha}/(1-\alpha)$, as $\alpha <1$.
Moreover, for $\alpha = \frac{2}{3} + \epsilon < 1$, as $\epsilon$ gets arbitrarily small
\[\sum_{i=1}^k i^{1-3\alpha} = \sum_{i=1}^k i^{-1-3\epsilon} 
\leq 1 + \int^\infty_1 x^{-1-3\epsilon} dx
= 1 + \frac{1}{3\epsilon}\,.
\]  

For $\epsilon$ small  and $k$ large enough, 
the LP cost is bounded by
$N^3(1+3\epsilon)/4\epsilon(1-3\epsilon) \approx N^3/4\epsilon$\,.
\end{proof}

\begin{lemma}\label{lem:igintcost} Any integral solution has cost at least $4N^3/(9\epsilon)$, up to lower order terms.
\end{lemma}
\begin{proof}
We will measure  the cost of an integral solution as the sum over time $t$ of the number of uncovered edges at $t$.
In any clique $K$, if $x$ vertices remain then there are exactly $x(x-1)/2 \approx x^2/2$ uncovered edges.
Consider the greedy algorithm, that at any time picks a vertex that covers the most number of edges.  
For our instance, it is easily verified that this in fact  minimizes the total number of remaining uncovered edges at every time $t$, and hence is clearly the optimum solution. The algorithm thus proceeds by first  covering vertices in $K_1$, until $n_2$ vertices are left {in $K_1$}, then alternately picking vertices from $K_1$ and $K_2$ until $n_3$ vertices are left in both, then alternating between $K_1,K_2,K_3$ until $n_4$ vertices are left and so on.

For $i \in [k]$, let $t_i$ denote the  first time when $n_i$ vertices are left in each of $K_1, \ldots, K_{i}$.
Then
$t_1=0$ and  
$t_{i+1} - t_i = i(n_{i}-n_{i+1})$
as the sizes of $K_1,\ldots,K_i$ must shrink from $n_i$ to $n_{i+1}$ during this time.

During $t_i$ to $t_{i+1}$, the number of uncovered edges is at least
$ i \cdot {n_{i+1} \choose 2}  + \sum_{j=i+1}^k {n_j \choose 2}$, where the first contribution is due to all the edges between the $n_{i+1}$ unpicked vertices in each of the first $i$ cliques $K_1,\ldots,K_i$, and the latter due to all the edges in the remaining cliques $K_{i+1},\ldots,K_k$. 
So the objective is at least
\begin{align}
\label{eq:MSVC-cost}
    \sum_{i=1}^{k-1} ( t_{i+1}-t_i)\left(i \cdot {n_{i+1} \choose 2}  + \sum_{j=i+1}^k {n_j \choose 2} \right) =  \sum_{i=1}^{k-1}  i(n_i - n_{i+1}) \left(i \cdot {n_{i+1} \choose 2}  + \sum_{j=i+1}^k {n_j \choose 2} \right)\,. 
\end{align} 
Note that we are even ignoring the contribution due to edges covered after time $t_k$.

Replacing ${n_i \choose 2} \approx \frac{1}{2} n_i^2$ we have
\[ i \cdot {n_{i+1} \choose 2}  + \sum_{j=i+1}^k {n_j \choose 2}   \approx \frac{N^2}{2}  \left(i^{1-2\alpha} + \int_i^k j^{-2\alpha} dj \right) =  
\frac{N^2}{2}  \left(
i^{1-2\alpha} + \frac{1}{2\alpha-1} \left(i^{1-2\alpha} - k^{1-2\alpha} \right)\right)\]
 Substituting this in  \eqref{eq:MSVC-cost} and writing $n_i - n_{i+1} \approx  N  \alpha i^{-1-\alpha}$ gives
\[ \frac{N^3}{2} \sum_{i=1}^k  i \cdot \alpha i^{-1-\alpha}  \left( \frac{2\alpha}{2\alpha-1} i^{1-2\alpha} - \frac{1}{2\alpha-1} k^{1-2\alpha}      \right)  = \frac{N^3}{2}  \left(\sum_{i=1}^k \frac{2\alpha^2}{2\alpha-1} i^{1-3\alpha}  -    \frac{\alpha k^{1-2\alpha}}{2\alpha-1}\sum_{i=1}^k i^{-\alpha}  \right) \]
As $\alpha > 2/3$, as $k$ gets large, the first term converges to $N^3 \alpha^2/(2\alpha-1)(3\alpha-2) \approx 4N^3/(9\epsilon)$. Next, as $\alpha < 1$, the second term is about $\frac{\alpha}{2\alpha-1}
\frac{k^{1-\alpha}}{1-\alpha} k^{1-2\alpha} \cdot (N^3/2) = O(k^{2-3\alpha}) \cdot (N^3/2) = o(N^3)$ for  $k=\omega(1)$. So any integral solution has cost at least $4N^3/(9\epsilon)$, up to lower order terms.
\end{proof}

Lemmas~\ref{lem:iglpcost} and~\ref{lem:igintcost}  give an integrality gap of $(4/9)/(1/4) = 16/9$.

\section{Conditioning for MSSC and MSVC}
\label{sec:conditioning}
We now prove Lemmas~\ref{lem:msscfin} and \ref{lem:cond-msvc}, that relate the expected cover time for an edge $e$ in $\sigma$ to $c_z(e)$ for MSSC and MSVC.  The conditioning is more subtle here, and handling it needs more care, and crucially uses the randomness in the tie-breaking rule (i.e., line $7$, Algorithm \ref{meta-algorithm}).

Fix some time $t$, and an edge $e$. Let $E_t$ denote the event that $\cov_\tau(e)=t$.
Recall that, given a solution $z$, we defined $c_z(e) = \sum_t p_t(e)$, where $p_t(e)$ was an upper bound on the probability that $e$ is uncovered just before $t$ in $\tau$, i.e.~$\Pr[\cov_\tau(e)  \geq t] \leq p_t(e)$.

Let us condition on the event $E_t$.
Then the expected cover time of $e$ in $\sigma$ is $1$ plus the delay due to vertices not in $e$.
Among the vertices not in $e$,
let $A_{<t}$ be those scheduled before $t$ in $\tau$, and $A_t$ be those scheduled at $t$. If $k$ vertices of $e$ are scheduled at $t$, then as the ties at time $t$ are broken at random, the expected delay of $e$ due to vertices in $A_t$ is exactly $|A_t|/(k+1) \leq |A_t|/2$ (as $k\geq 1$ conditioned on $E_t)$.

So the expected cover time of $e$ in $\sigma$ conditioned on $E_t$ is at most $1 + \expec{}{|A_{< t}|} + \expec{}{|A_t|/2}$.
Now, any $v \notin e$ is scheduled before $t$ in $\tau$ with probability $z_{v,<t}$, and this is independent of $E_t$. So,
\[ \expec{}{|A_{< t}|}  = \sum_{v \notin e} z_{v,<t} = \beta(t-1) - \sum_{v \in e} z_{v,<t} 
 \leq \beta (t-1) - z_{e, <t }\,.\] Here we used 
$z_{e,t} \leq \sum_{v\in e} z_{v,t}$ for MSSC and MSVC. 
 Similarly, $\expec{}{|A_{t}|}  \leq \beta - z_{e,t}$. Hence,
\begin{align}
    \expec{}{\cov_\sigma(e)} & \leq \sum_t \Pr[E_t] \left(1 + \beta(t-1) - z_{e,<t} + (\beta - z_{e,t})/2 \right) \nonumber \\
    & \leq \sum_t  (p_t(e)-p_{t+1}(e) ) (1 + \beta(t-1) - z_{e,<t} + (\beta - z_{e,t})/2), \label{cond:ineq} 
\end{align}  
where we use that as $(\beta(t-1) - z_{e,<t} + (\beta - z_{e,t})/2)$ is non-decreasing in $t$ and $\Pr[\cov_\tau(e)  \geq t] \leq p_t(e)$,  replacing $\Pr[E_t]$ by $p_{t}(e)- p_{t+1}(e)$ can only increase the right hand side.

As $c_z(e) = \sum_t p_t(e) = \sum_t t(p_t(e)-p_{t+1}(e))$,
to prove 
$ \expec{}{\cov_\sigma(e)} \leq \beta c_z(e)$, it suffices to show that the expression in \eqref{cond:ineq} is at most  $\beta \sum_t t(p_t(e)-p_{t+1}(e))$. 

Using $\sum_t (p_t(e) - p_{t+1}(e))=1$ and
simplifying, this reduces to showing that
\[1 - \beta/2  \leq   \sum_t  (p_t(e)-p_{t+1}(e) ) (z_{e,<t} + z_{e,t}/2).\]
As $\beta \geq 1 $, the left side is at most $1/2$, and the right hand side simplifies to
$\sum_t z_{e,t} (p_{t}(e) + p_{t+1}(e))/2.$
So our goal is to show that
\[\sum_t z_{e,t} (p_{t}(e) + p_{t+1}(e)) \geq 1.\]
We now do this separately for MSSC and MSVC.

\noindent{\em MSSC:} In the definition of $c_z(e)$, we defined $p_t(e) = \exp(-z_{e,<t})$, so our goal is to show that
\beq \label{eq:abc} \sum_t \big(\exp(-z_{e,<t}) + \exp(-z_{e,<t+1})\big) z_{e,t} \geq 1. \eeq
Now, for any convex function $f$ and any two points $a$ and $b$ it holds that  $(f(a)+  f(b)) (b-a) \geq 2 \int_a^b f(u) du$. 
Setting $f(u)= \exp(-u)$ and $a= z_{e,<t}$ and $b=z_{e,<t+1}$ for each $t$, the left hand side in \eqref{eq:abc} least 
\[ 2 \sum_t \int_{z_{e,<t}}^{z_{e,t+1}} \exp(-w)dw = 2 \exp(-z_{e,<1}) -2 \exp(-z_{e,<\infty}) =2 > 1,\]
which proves the desired result for MSSC (Lemma~\ref{lem:msscfin}).

\noindent{ \em MSVC.} Here, in the definition of $c_z(e)$ we defined $p_t(e) = (1-z_{e,<t}/2)^2$, which is convex for $z_{e,<t} \in [0,2]$. For MSVC, $e$ is surely covered if $z_{e,<t} \geq 2$, and hence $p_t(e)=0$ for such $t$. So we only need to consider $t$, where $z_{e,<t} \leq 2$. So, following the argument above,
\[ \sum_t z_{e,t} (p_{t}(e) + p_{t+1}(e)) \geq 2 \sum_t \int_{z_{e,<t}}^{z_{e,<t+1}} (1-w/2)^2 dw = 2 \int_0^2 (1-w/2)^2 dw = 4/3 > 1.\]
This proves the desired inequality for MSVC (Lemma~\ref{lem:cond-msvc}).

\section{A Tail Bound for Sum of Bernoulli R.V.s}
\label{sec:tail}
In this section we present our main ideas for proving Theorem \ref{thm:enhanced-bound}.
Recall that we are going to upper bound $Pr[S \leq k-1]$ where $S$ is sum of $n$ independent Bernoulli random variables, having $\expec{}{S} = k \gamma$ for $\gamma \geq 1$ and integer $k \geq 1.$

First, we observe assuming our Bernoulli variables are identical and increasing their number, i.e., $n$, can only increase the left tail we are going to upper bound. This allows to reduce the problem to bounding the left tail of a Poisson distribution. Next, by an induction on $k$ we reduce the problem to the case when $\gamma = 1$.

The following lemma allows us to reduce the space of distributions under study to sum of independent and \emph{identical} Bernoulli random variables, i.e., a Binomial distribution $B(n, p)$, for which we denote the cumulative distribution function (CDF) by
$
F(k;n,p) \defeq \sum_{i = 0}^{k} {n \choose i} p^i (1-p)^{n-i}\,,
$ for $k \in Z^+$\,.

\begin{lemma}[\cite{hoeffding1956distribution}, Theorem 4]\label{lem:Hoeffding}
Among all choices of $p_i$ with $\sum_i p_i = \mu$ fixed, for $k \leq \mu$, \  $\Pr[S \leq k-1]$ is maximized when all the $p_i$'s are identical, equal to $p=\mu/n$ .
That is, 
$$
\Pr[S \leq  k-1] \leq F(k-1;n,p) = \sum_{i = 0}^{k-1} {n \choose i} p^i (1-p)^{n-i}.$$
\end{lemma}

Now we apply the above Hoeffding's lemma again to reduce our problem to bounding CDF of a Poisson distribution.

\begin{lemma}\label{lem:PoissonBound}
For $\lambda = np$,
$
F(k-1;n,p) \leq e^{-\lambda} \sum_{i=0}^{k-1} \frac{\lambda^i}{i!}\,.
$
\end{lemma}

\begin{proof}
Define $S^{\prime} \defeq \sum_{i \in [2n]} X_i^{\prime}$ where
$$
X_i^{\prime} \sim \left\{\begin{array}{lr}
        Bernoulli(p), & i \in [n]\\
        Bernoulli(0), & i \in [2n] \setminus [n] \\
\end{array}\right.\,.
$$

Given that $\expec{}{S^{\prime}} = np$, the previous lemma gives $\Pr[S' \le k-1] \le F(k-1;2n,p/2)$. On the other hand, $\Pr[S^{\prime} \le k-1] = F(k-1;n,p)$. Thus we have $F(k-1;n,p) \leq F(k-1;2n,p/2)$, and more generally, $F(k-1;n,p) \leq \lim_{\alpha \rightarrow \infty} F(k-1;\alpha n, p/\alpha)$. Now the desired bound in the lemma follows from the convergence of the binomial distribution to the Poisson distribution, with $\lambda = np$, as $n \rightarrow \infty$.
\end{proof}

Recall that $\gamma = np/k = \lambda / k$. We are going to reduce the problem to the case $\gamma = 1$. First see that for $\gamma = 1$ (i.e.,  $\lambda = k$), applying  Lemmas \ref{lem:Hoeffding} and \ref{lem:PoissonBound} reduces the proof of Theorem \ref{thm:enhanced-bound} to showing 
\begin{equation}
    \label{one-half-bd}
e^{-k} \sum_{i=0}^{k-1} k^i/i! \leq e^{-1}(e/2) = 1/2\,.
\end{equation}
The left hand side is the probability that a Poisson random variable with mean $k$ has value at most $k-1$. We have the desired bound using the fact that a median of the Poisson distribution with mean $\lambda$ is strictly larger than $\lambda -1$; see \cite{JS68} for a precise estimate. Thus it remains to prove Theorem~\ref{thm:enhanced-bound} for $\gamma > 1$. Here we present the main idea by proving a weaker bound. The actual proof of Theorem~\ref{thm:enhanced-bound} is deferred to Appendix \ref{sec:enhanced}.

\begin{theorem}\label{thm:easytail}
Let $S = \sum_i X_i$, where $\set{X_i \sim \text{Bernoulli}(p_i): i \in [n]}$ are independent Bernoulli random variables, with $p_i \in [0,1]$. Let $k \leq \expec{}{S}$ be a positive integer. For  $\gamma = \expec{}{S} / k \ge 1$, we have 
\begin{align}
\Pr[S \leq k-1] \le e^{-\lambda} \sum_{i=0}^{k-1} \frac{\lambda^i}{i!} \leq e^{-\gamma}(e/2)\,.
\end{align}
\end{theorem}
\begin{proof}
Let $g(\gamma) := e^{-\gamma}(e/2)$ denote the desired upper bound.
Define
\[f_k(\lambda) \defeq
e^{\lambda} g(\lambda/k)  - \sum_{i=0}^{k-1} \frac{\lambda^i}{i!}.\]
With this notation, the theorem is equivalent to showing that $f_k(\lambda) \ge 0$ for $\lambda \ge k$. We prove this by induction on $k\ge 1$.
For $k = 1$, the basis, this holds trivially as
$ f_1(\lambda) = e^{\lambda} g(\lambda) -1 = e/2-1 \geq 0.$

Fix $k \ge 2$. For $\lambda = k$ we previously showed that $f_k(k) \ge 0$ (which is equivalent to \eqref{one-half-bd}), so it suffices to prove $f_k(\cdot)$ is a non-decreasing function, i.e., 
$\frac{d}{d \lambda} f_k(\lambda) \ge 0$, for $\lambda \ge k$.
By the inductive hypothesis, we have $f_{k-1}(\lambda)\geq 0$, for all $\lambda \geq k-1$; that is,
\begin{equation}
\label{eq:exppartial}
    \sum_{i=0}^{k-2} \frac{\lambda^i}{i!} \leq e^{\lambda} g(\lambda/(k-1))\,.
\end{equation}
For the induction step we have
\begin{align*}
\frac{d}{d \lambda} f_k(\lambda) &=
e^{\lambda}\left(\frac{1}{k} g'({\lambda}/{k}) + g({\lambda}/{k})\right) - \sum_{i=1}^{k-1} \frac{i \lambda^{i-1}}{i!}
=e^{\lambda} \Bigl(\frac{1}{k} g'({\lambda}/{k}) + g({\lambda}/{k})\Bigr) - \sum_{i=0}^{k-2} \frac{ \lambda^{i}}{i!}\\
&\geq e^{\lambda} \Bigl(\frac{1}{k} g'({\lambda}/{k}) + g({\lambda}/{k}) - g(\lambda/{k-1})\Bigr)& (\text{by } \eqref{eq:exppartial})
\end{align*}
Substituting $\gamma = \lambda/k$ it suffices to show for all $k\geq 2,\gamma \geq 1$
\beq
\frac{g'(\gamma)}{k} + g(\gamma) - g\left(\frac{\gamma k}{k-1}\right) \geq 0 
\eeq
For $g(\gamma) = \exp(1-\gamma)/2$, the above becomes
\[ 
-\frac{\exp(-\gamma)}{k} + \exp(-\gamma) - \exp\Bigl(-\frac{\gamma k}{k-1}\Bigr) \geq 0
\]
The proof now follows by re-writing the LHS of the above as
\[\frac{k-1}{k}\exp\Bigl(-\frac{\gamma k}{k-1}\Bigr)\left(\exp\Bigl(\frac{\gamma}{k-1}\Bigr) -\frac{k}{k-1}\right)\geq 0\]
where the last inequality is due to $\gamma\geq 1$ and $\exp(x) \geq 1+x$ for all $x$.
\end{proof}

\section{Min $\ell_p$ Set Cover}\label{sec:minLPsetcover}
We now consider the problem of minimizing the $\ell_p$ norm of cover times for MSSC. We refer to this as MSSC$_p$. 
We first sketch how the previous approach gives a $\delta_p = (p+1)^{1+1/p}$ approximation. Then we focus on the greedy algorithm and show that it simultaneously provides this guarantee for every $p\geq 1$. We then show a hardness result that no better approximation exists for any $p\geq 1$, assuming P$\neq$NP.

\subsection{A Randomized Rounding Algorithm}

Consider the the natural LP relaxation for MSSC$_p$, which is similar to the LP for MSSC but with a different objective, \[ \sum_t (t^p - (t-1)^p) u_{e,t}\,, \] 
i.e., the cost is the $p$th power of the $\ell_p$ norm of the cover times.
We show the following.
\begin{theorem}\label{thm:lpnormrand}
The rounding algorithm in Section \ref{sec:GMSSC} with $\beta=p+1$
for the above LP
is a $\delta_p$ approximation for MSSC$_p$.
\end{theorem}
\begin{proof}
The analysis is similar to that of Theorem~\ref{thm:mssc} for MSSC.
Let us define
\[ c_z(e) = \sum_t (t^p - (t-1)^p) \exp(-\sum_{v \in e} z_{v,<t}) \quad\text{ and }\quad 
c_x(e) = \sum_t t^p x_{e,t},
\]
where $c_z(e)$ is an upper bound on the expected cost of $e$ in the tentative schedule, and $c_x(e)$ is the contribution in the LP.
We first show that $c_z(e) \leq (p+1) c_x(e)$. Similar to Theorem~\ref{thm:mssc}, it suffices to upper bound the value of the following optimization problem in the variables $a_t$.
\[ \qquad  \text{Maximize }  \frac{\sum_t (t^p - (t-1)^p)\exp\Bigl(-(p+1)\sum_{t^{\prime}< t} (a_{t'} \ln\frac{ t}{t^{\prime}})\Bigr)}{\sum_t t^p a_t} \qquad \text{ s.t. }\  \|a\|_1 =1, \qquad a_t \geq 0\,.  \]
Similar to Claim~\ref{cl:optmssc}, the optimum is at an extreme point supported on $a_u$ for some $u$. The denominator is $u^p$ and the numerator is
\begin{align*}
 & \sum_{t=1}^u  (t^p - (t-1)^p)  +  \sum_{t > u}  (t^p - (t-1)^p)\exp (-  (p+1) \ln t/u)  \\ 
 \leq & \, u^p  +  \sum_{t > u} p t^{p-1}(u/t)^{p+1}
\leq  u^p + pu^{p+1}\int^\infty_{x=u} \frac{1}{x^2} dx
= u^p(1 + p).
\end{align*}

Hence, $c_z(e) \leq (p+1) c_x(e)$. As $\beta=p+1$, converting the tentative schedule to an actual schedule will increase the expected cover time by at most $p+1$, and its $p$-th power by at most $(p+1)^p$. By an argument similar to that in Section~\ref{sec:conditioning} to account for conditioning, the expected cost of the final solution is at most $(p+1)^p$ times $\sum_e c_z(e)$ which is at most $(p+1)^{p+1} \sum_e c_x(e)$. This gives the claimed $\delta_p$ approximation.
\end{proof}

\subsection{Bound for the Greedy Algorithm}
The greedy algorithm, at every time-slot $t=1,2,3,\ldots$ schedules a vertex $v$ that covers the most number of currently uncovered edges. Let $g$ denote the cost of the algorithm and OPT denote the optimum value.

\begin{theorem}
The greedy algorithm gives a
$\delta_p$-approximation for MSSC$_p$, for every $p\geq 1$.
\end{theorem}
We prove this using a dual fitting argument. Consider the natural LP formulation for minimizing the $p$-th power of the $\ell_p$ norm. We use continuous time for convenience (note that this can only reduce the LP objective as compared with discrete time).
\begin{align*}
(\text{LP}_p) \quad \text{Min} \quad  \displaystyle\sum\limits_e \int_0^\infty p t^{p-1} u_{e,t} dt & 
\quad\text{s.t.} \quad
\displaystyle \sum_v y_{v,t} \leq t\quad \forall~ t,\quad
u_{e,t} + \sum_{v \in e}y_{v,t} \ge 1 \quad\forall\ e,t,\ \ 
u_{e,t},y_{v,t} \ge 0\,. 
\end{align*}
Here the variable $y_{v,t}$ indicates whether a vertex is scheduled by time $t$ and $u_{e,t}$ indicates whether an edge $e$ is not covered by time $t$. The objective function for an edge $e$ now becomes $\int_0^\infty p t^{p-1} u_{e,t} dt$ (which is $t^p_e$ if the edge $e$ is covered at time $t_e$ in an integral schedule).

 Consider a new LP, namely LPS$_p$, obtained by replacing the inequality $\sum_v y_{v,t} \leq t$ in  LP$_p$ by
\[\displaystyle \sum_v y_{v,t}  \leq t/(p+1)   \quad\forall~ t\,,\]
which can be interpreted as a simple time-scaling.
We have the following simple relation.
\begin{lemma}\label{lem:scaling}
$\text{Opt}(LPS_p) =(p+1)^p \cdot \text{Opt}(LP_p)$.
\end{lemma}

\begin{proof}
Given an optimum solution $u_{e,t},y_{v,t}$ to LP$_p$, let $y'_{v,t} = y_{v,t/(p+1)}$ and $u'_{e,t} = u_{e,t/(p+1)}$.
It is easy to verify that $u'_{e,t},y'_{v,t}$ is a feasible solution for LPS$_p$. It suffices to upper bound the cost of $\text{LPS}_p$ for $(u',y')$ by $(p+1)^p$ times the cost of $LP_p$ for $(u,y)$. Indeed, by a simple change of variable in the integration, we have
\[\sum_e \int_0^\infty p t^{p-1} u'_{e,t} dt =\sum_e \int_0^\infty p t^{p-1} u_{e,t/(p+1)} dt = (p+1)^p\sum_e \int_0^\infty p t^{p-1} u_{e,t} dt\,. \qedhere\]
\end{proof}

Let us consider the dual to LPS$_p$, obtained using variables $\alpha_t$ (resp.~$\beta_{t,e}$) for the first (resp.~second) set of constraints. The dual is as follows.
\begin{align*}
(\text{DLPS}_p) \qquad \text{Max} \quad  &\displaystyle\sum\limits_e \int_0^\infty   \beta_{e,t}dt - \frac{1}{p+1} \displaystyle\int_0^\infty   {t\cdot \alpha_{t}} dt \\   \text{s.t.} \quad 
&\displaystyle\sum\limits_{e \ni v} \beta_{e,t}  \leq \alpha_t \quad \forall~ v,t, \quad
0 \leq \beta_{e,t}  \leq p t^{p-1}    \quad \forall~ e,t, 
\quad \alpha_t \ge 0 \quad  \forall t\,.  
\end{align*}

In the following key lemma, we use the execution of the greedy algorithm to construct a feasible dual solution. 
\begin{lemma}\label{lem:greedycost}
There is a feasible solution to the dual LP above with cost at least ${g^p}/({p+1})$.
\end{lemma}
\begin{proof}
Consider the execution of the greedy algorithm. Note that greedy schedules at discrete time steps, while the LP above is for continuous time. Let us denote the set of edges not covered at the beginning of time-slot $k$ by $R_k$, and the edges with cover time $k$ (i.e., $R_{k+1} \setminus R_{k}$) by $X_k$ respectively. In other words, $X_k$ is exactly the set of edges covered at time $k$ by the greedy algorithm.

Consider the solution $\alpha_t = pt^{p-1}|X_{\lceil t \rceil}|\,$ and $\beta_{e,t} = pt^{p-1}\mathbbm{1}[e \in R_{\lceil t \rceil}],$ for any time $t \in \mathbb{R}^+$. We claim that this is feasible.
Feasibility of the dual constraint $\beta_{e,t} \leq pt^{p-1}$ is easy to see by construction. Let us consider the other constraint $\sum_{e \ni v} \beta_{e,t} \leq \alpha_t$ for some vertex $v$ and time $t$. Then,
\[\sum_{e \ni v} \beta_{e,t} = p t^{p-1}\sum_{e \ni v} \mathbbm{1}[e \in R_{\lceil t \rceil}] \leq pt^{p-1} |X_{\lceil t \rceil}| = \alpha_t.\]
Here the inequality uses the key property of the greedy algorithm that among all vertices $v$, and at all times $t$, the vertex 
chosen at time $\lceil t \rceil$ by the greedy algorithm is the one that maximizes
$
\sum_{e \ni v} \mathbbm{1}[e \in R_{\lceil t \rceil}]$.

\noindent \textbf{Cost analysis.} We now bound the cost of the dual solution.
Let $t_e$ denote the cover time of an edge $e$ under the greedy algorithm. Then we have the following.
\[
\sum\limits_e \int_0^\infty   \beta_{e,t}dt = \sum_e  \int^{t_e}_0 p t^{p-1} dt = \sum_e t^p_e = g^p,\]
\[
\int_0^\infty t \alpha_t  dt = \sum_{i=1}^\infty \int_{i-1}^i t \alpha_t dt = \sum_{i=1}^\infty |X_i| \int_{i-1}^{i} t \cdot pt^{p-1} dt
\leq 
\sum_{i=1}^\infty |X_i| p i^p
=
p \sum_e t_e^p = p \cdot g^p\,.
\]
The second to last inequality follows as the greedy objective can be written in two different ways as $\sum_i |X_i| i^p$ and $\sum_e t_e^p$.

Together, these  give that the dual cost is at least
$g^p - \frac{p}{p+1} g^p =  g^p/(p+1)$ as desired.
\end{proof}

By Lemmas~\ref{lem:greedycost} and \ref{lem:scaling}, we have that
\[ g^p/(p+1) \leq \text{Opt}(\text{DLPS}_p) \leq \text{Opt}(\text{LPS}_p) \leq (p+1)^p \cdot \text{Opt}(\text{LP}_p) \leq (p+1)^p \cdot \text{OPT}^p\,,
\]
where the second inequality is by weak duality. Hence, $g^p \leq (p+1)^{p+1} \cdot \text{OPT}^p$ completing the proof.

\subsection{Hardness of Approximation}
We now show the following hardness result.
\begin{theorem}\label{thm:hardness}
It is NP-hard to approximate MSSC$_p$ within a ratio of $\delta_p - \epsilon$, for any $\epsilon > 0$ and $p \geq 1$.
\end{theorem}
Our approach is based on that of \cite{FLT04} to show the $4-\epsilon$ hardness of MSSC, with some change in parameters. The starting point is the following result of \cite{FLT04,F98}.
\begin{theorem}\label{thm:feigereg}
Consider the problem of picking vertices to cover edges in a regular uniform hypergraph.
For any $c_0, \epsilon>0$, 
it is NP-hard to distinguish between 

\textbf{Case I:} All edges can be covered by an independent set of vertices, of size at most $r$.

\textbf{Case II:} For any $c \leq c_0 r$, every subset of $c$ vertices covers at most $1-(1-1/r)^c + \epsilon$ fraction of edges.
\end{theorem}

Let $H$ be a regular uniform hypergraph with $m$ edges. Fix a large constant $k$ and construct $k$ disjoint copies of $H$, indexing them as $1,2,\ldots,k$. Make ${a}/({i^{p+1}})$ copies of every edge in the $i$th copy of $H$ where $a = (k!)^{p+1}$, denote the resulting hypergraph by $H_i$. Consider the instance $G = \cup_i H_i$.

Let us call the cost of an edge to be the $p$th power of its cover time. Then,
it suffices to show that the total cost for $G$ in case II is at least $(p+1)^{p+1}$ times the cost in case I.

\textbf{Assuming Case I for $H$.}
Let $F(i)$ be the contribution of $H_i$ to the cost. $F(i)$ can be upper bounded by noting that each additional vertex scheduled from $H_i$ covers $ma/(ri^{p+1})$ edges of $H_i$.
\[
F(i) \leq \frac{a}{i^{p+1}} \sum_{j = 1}^{r} \left(\frac{m}{r}\right) \left((i-1)r + j\right)^p
\approx
\frac{ma}{(i+1)^p \cdot r} \cdot \frac{(ir)^{p+1} - ((i-1)r)^{p+1}}{(p+1)r}
\approx
\frac{mar^p}{i}\,.
\]
where $r,i$ are assumed to be sufficiently large.
Hence,
$ \sum_i F(i) \lessapprox mar^p \log k. $

\textbf{Assuming Case II for $H$.} We know that no matter which $c < c_0 r$ vertices in $H_i$ are picked, at least $(1-1/r)^c$ of the edges remain uncovered. Let $m_i = m a/(i^{p+1})$ be the number of edges in $H_i$. Now, the best way to cover $G$ is to first schedule vertices from $H_1$ until $\leq m_2$ edges remain uncovered in $H_1$, then schedule vertices from $H_1$ and $H_2$ alternately until $\leq m_3$ edges remain in $H_1,H_2$, and so on. It can be verified that in our case, this minimizes the total number of remaining uncovered edges at every time step, and hence is the optimum solution.

Let $t_i$ denote the first time when $m_i$ edges are left to be covered in $H_1,\ldots,H_i$. Then $t_1 = 0$ and $t_{j+1} - t_j \approx j r \ln (m_j/m_{j+1})$ because during $(t_{j},t_{j+1}]$, the number of uncovered edges in any of $H_1,H_2,\ldots,H_{j}$ decreases to a $(1-1/r)^{x/j} \approx \exp(-x/jr)$ fraction of its value at $t_{j}$ every $x$ steps (since each of the first $j$ copies has a vertex scheduled from it every $j$ steps). This gives
\[ t_i \approx  r\sum^{i-1}_{j=1} j\ln \frac{m_{j}}{m_{j+1}} = (p+1)r \sum^{i-1}_{j=1} j(\ln (j+1) - \ln j) = (p+1)r \Bigl((i-1) \ln i - \sum^{i-1}_{j=1} \ln j\Bigr) \approx (p+1) ri \]
Consider the $i (m_i - m_{i+1})$ edges covered between $t_i$ and $t_{i+1}$. The cost of each of them is at least $t^p_i \approx ((p+1)ri)^p$. Hence, the total cost is at least 
\[ \sum^k_{i=1} i t^p_i (m_i - m_{i+1})
\approx i(p+1)^p r^p i^p\sum^k_{i=1} i^{p+1} r^p ma \frac{(p+1)}{i^{p+2}} = (p+1)^{p+1} mar^p \sum^k_{i=1} \frac{1}{i} 
\approx (p+1)^{p+1} mar^p \ln k 
\]
which is $(p+1)^{p+1}$ times $\sum^k_i F(i)$. Hence it is NP-hard to approximate MSSC$_p$ within $\delta_p - \epsilon$.

\section{The conjecture of \cite{ISV14} on rounding preemptive solutions}\label{sec:imconj}
Im et al.~\cite{ISV14} defined the notion of a \emph{preemptive} schedule as an assignment $x_v(t)$  of vertices $v$ to real-valued times $t$ satisfying $\sum_v x_v(t)= 1$ for all $t$ and $\int^\infty_0 x_v(t) dt = 1$ for all $v$. The cover time of an edge $e$ in the preemptive schedule is defined to be the (real-valued) time $t$ such that $\int^t_0 \sum_{v \in e} x_v(t) dt = k_e$.

Im et al.~\cite{ISV14} conjectured that for any instance there always exists a non-preemptive schedule (which is our usual notion of a schedule) of cost at most twice the cost of any preemptive schedule. 

We now give a counter-example, which shows that this gap must be at least $4$. This holds even for MSSC, and the instance we use is the same as that used by \cite{FLT04} for their hardness result. As the approach of \cite{ISV14} loses another factor $2$ to obtain a preemptive schedule from the configuration LP, this implies that their approach cannot give a better than $8$ approximation (even for MSSC).

We recall that Feige et al.~\cite{FLT04,F98} showed the following.
\begin{claim}\label{cl:hardinst}
For every $c_0,\epsilon>0$, there exists a uniform regular hypergraph $H$ with every edge having $|V|/r$ vertices (for sufficiently large $r$ dependent on $c_0$) such that any collection of $c \leq c_0$ vertices covers at most $1-(1-1/r)^c+\epsilon$ fraction of the edges.
\end{claim}

Using $H$ from \Cref{cl:hardinst}, consider the following hard instance as in~\cite{FLT04}. Make $k$ disjoint copies of $H$ (for some large enough constant $k$), on distinct universe of elements. For each $1 \leq i\leq k$, make $a/i^2$ copies of every edge in the $i$th copy of $H$ where $a = (k!)^2$, call the resulting hypergraph $H_i$. Note that $H_i$ has $ma/i^2$ edges and $n$ vertices where $m = |E(H)|,n=|V(H)|$. 

Feige et al.~\cite{FLT04} showed that for this instance any (non-preemptive) schedule must have cost $\approx 4 mar\ln k$. So it suffices to exhibit a preemptive schedule with cost at most $mar \ln k$.

Consider the following preemptive schedule: For $i=1,2,3\ldots$, schedule the vertices of $H_i$ uniformly for $r$ time units i.e. $x_{v}(t) = 1/n$ for $v \in H_i,t \in [r(i-1),ri]$. Each edge $e \in H_i$ is covered by time $ri$ because the $n/r$ vertices of $e$ have each been scheduled to an extent $r/n$. The cost of the preemptive schedule is $\sum^k_{i=1} (ma/i^2)(ri) \approx mar \ln k$ as claimed.

\bibliographystyle{plain}
\bibliography{refs}

\appendix

\section{Proof of Claim \ref{cl:maxa1}}
\label{sec:maxa1}

Let $f(a)$ and $g(a)$ be defined as 
\[ f(a) = \left( 1 + \int^a_1 P(1+s \beta \ln x)dx + \int^\infty_a P(1+\beta \ln x - \beta (1-s) \ln a\right)\]
\[ g(a)  = s e^{-1/(s\beta)} + a(1-s).\]
\begin{claim}
For any $s \in [0,1)$ and $\beta \geq 1$, the ratio $f(a)/g(a)$, for $a \geq 1$, is maximized at $a=1$.
\end{claim}

\begin{proof}
To prove that $f(a)/g(a)$ is maximized at $a=1$, it suffices to show that $(f(a)/g(a))' \leq 0$ for all $a \geq 1$. To this end, we will show that $h(a) = f'(a)g(a) - g'(a)f(a)$ is negative for all $a \geq 1$. We do this by proving that $h(1) \leq 0$ and $h'(a) \leq 0$ for $a \geq 1$.

Let us first note that $g'(a) = 1-s$ and $g''(a)=0$. We now
 compute $f'$. 
\begin{align*}
f'(a) &= P(1+s\beta \ln a) - P(1+\beta \ln a - \beta(1-s) \ln a) - \int^\infty_a \beta (1-s) \frac{1}{a} P'(1 + \beta \ln x - \beta (1-s) \ln a) dx\\
& = - \int^\infty_a \beta (1-s) \frac{1}{a} P'(1 + \beta \ln x - \beta (1-s) \ln a) dx \\
&= - \int^\infty_1 \beta (1-s) P'(1 + s \beta \ln a + \beta \ln t) dt
\end{align*}
where the third equality follows by substituting $x = at$.

\noindent {\em Proving $h'(a) \leq 0 \quad \forall a \geq 1$.}
As $h' = f''g - g''f$, and $g''=0$, $g(a) \geq 0$, it suffices to show that $f''(a) \leq 0$ for all $a \geq 1$. 
Using the expression for $f'(a)$ above, gives
\[f''(a) = - \int^\infty_1 \beta^2 s(1-s) \frac{1}{a} P''(1 + s \beta \ln a + \beta\ln t) dt \]
As $P$ is convex in $[1,\infty)$, $P''(1+s\beta \ln a + \beta \ln t) \geq 0$ for all $t\geq 1$ and hence $f''(a) \leq 0$.

It remains to show that $h(1) \leq 0$.
Let us compute $f(1),g(1),f'(1),g'(1)$. We have
\[f(1) = 1 + \int^\infty_1 P(1+\beta \ln x) dx, \qquad 
f'(1) = - \int^\infty_1 \beta (1-s) P'(1+\beta \ln x) dx.\]
Next, $g(1) = s e^{-1/s\beta} + (1-s)$ and $g'(1) = 1-s$.

Noting that $f'(1) \geq 0$ (as $P$ is decreasing in $[1,\infty)$ and hence $P'(x) \leq 0$ for all $x \geq 1$) and $g(1) \leq 1$ (using $\exp(-1/s\beta) \leq 1$), we have that
\begin{align*}
 h(1) &= f'(1)g(1) - f(1)g'(1) \leq f'(1) - f(1)(1-s)\\
 &\leq -(1-s) \left(1 + \int^\infty_1 \beta P'(1 + \beta \ln x) + P(1 + \beta \ln x) dx \right)\\
 &=-(1-s)\left(1+\int^\infty_1 e^{\frac{y-1}{\beta}} \Bigl( P'(y) + \frac{P(y)}{\beta}\Bigr) dy\right)\\
 &=-(1-s)\left(1+\int^\infty_1 \Bigl(e^{\frac{y-1}{\beta}}P(y)\Bigr)' dy \right) \leq -(1-s)(1 - P(1)) \leq 0\,,
 \end{align*}
where the equality in the third line is by substituting $1+\beta \ln x = y$.
\end{proof}

\section{Proof of Theorem \ref{thm:enhanced-bound}}\label{sec:enhanced}

Following the framework of the proof of Theorem~\ref{thm:easytail}, it suffices to show that for $g(\gamma) = e^{-\gamma}(e/2)^{\exp(1-\gamma)}$,
\begin{align}\label{suff-for-induction}
\frac{1}{k}  g'(\gamma)  + g(\gamma) - g\Bigl(\gamma \frac{k}{k-1} \Bigr) \geq 0, \quad\text{ for all } \gamma \geq 1,\ k \geq 2\,.
\end{align}
Substituting $f(\gamma) = \ln g(\gamma)$ and noting $g'(\gamma) = f'(\gamma)g(\gamma)$,~\eqref{suff-for-induction} becomes
\[ \frac{1}{k}f'(\gamma) + 1 \geq  \frac{g\Bigl(\gamma \frac{k}{k-1}\Bigr)}{g(\gamma)}\,. \]
Taking the logarithm on both sides,
\[ \ln \Bigl (\frac{1}{k}f'(\gamma)+1\Bigr) \geq f\Bigl(\gamma \frac{k}{k-1}\Bigr) - f(\gamma)\,. \]
Let $\theta$ denote $\ln(e/2)$. Plugging $f(\gamma) = -\gamma + e\theta e^{-\gamma}$ and $f'(\gamma) = - 1 - e\theta e^{-\gamma}$, it suffices to show

\beq\label{eq:maineq} \ln \Bigl(1-\frac{1}{k}-\frac{e\theta}{k}e^{-\gamma}\Bigr) \geq - \frac{\gamma}{k-1} + e\theta\left(\exp\Bigl(-\frac{\gamma k}{k-1}\Bigr)-e^{-\gamma}\right)\,. \eeq
Denoting the left hand side  of~\eqref{eq:maineq} by $\ell_k(\gamma)$ and the right hand side by $r_k(\gamma)$, we first show the following.
\begin{claim}\label{cl:incdec} 
For any fixed $k\geq 2$, $\ell'_k(\gamma) \geq 0$ and $r'_k(\gamma)\leq 0$.
\end{claim}
\begin{proof}
Since $\ln$ and $-e^{-\gamma}$ are non-decreasing and $\theta>0$, $\ell'_k(\gamma)\geq 0$. Also,
\begin{align*}
r'_k(\gamma) &= \frac{-1}{k-1}+e\theta\left(\frac{-k}{k-1}\exp\Bigl(\frac{-\gamma k}{k-1}\Bigr) + e^{-\gamma}\right)
=e\theta e^{-\gamma}\left(\frac{-e^\gamma}{e\theta(k-1)} -\frac{k}{k-1}\exp\Bigl(\frac{-\gamma}{k-1}\Bigr) + 1\right)\\
&\leq e^{-\gamma}\left(-\frac{e^\gamma}{k-1} - \exp\Bigl(-\frac{\gamma}{k-1}\Bigr)+1\right)
\leq e^{-\gamma}\left(-\frac{e^\gamma}{k-1} + \frac{\gamma}{k-1}\right)
\leq 0\,,
\end{align*}
where the second and third inequalities are by using $e^{-x} \geq 1-x$ for $x \geq 0$ and $e\theta < 1$ respectively.
\end{proof}

By Claim~\ref{cl:incdec}, it suffices to prove~\eqref{eq:maineq} for $\gamma=1$. That is, for all $k \geq 2$,
\beq\label{eq:desired}
\ln \Bigl( 1- \frac{1}{k}  - \frac{\theta}{k} \Bigr)        \geq  -\frac{1}{k-1}  + \theta\left(\exp\Bigl(-\frac{1}{k-1} \Bigr) -1\right)\,. \eeq 
Substitute $x = 1/(k-1)$. As $k\geq 2$, $x \in [0,1]$ and $1/k = x/(1+x)$. Then by ~\eqref{eq:desired}, it suffices to show that
\[\ln \Bigl( \frac{1-\theta x}{x+1} \Bigr) + x - \theta (e^{-x}-1) \geq 0\,. \]
Let us denote the LHS by $h(x)$ and note that $h(0) = 0$. It suffices to show $h'(x) \geq 0$ for $x\in [0,1]$.
\begin{align*}
h'(x) &= \frac{-\theta}{1-\theta x} - \frac{1}{x+1} + 1 + \theta e^{-x}
\geq \frac{-\theta}{1-\theta x} - \frac{1}{x+1} + 1 + \theta (1 - x) = \frac{x(1-\theta-\theta^2-2\theta x + \theta^2 x^2)}{(1-\theta x)(x+1)}\,,
\end{align*}
once again using, $e^{-x} \geq 1-x$.

It suffices to show that $q(x) = (1-\theta-\theta^2-2\theta x + \theta^2 x^2) \geq 0$, since the other terms are non-negative. $q'(x) = 2\theta(-1+\theta x) \leq 0$ for $x \in [0,1]$ since $\theta = 1 - \ln 2 < 1$. Noting that $q(1) = 1 - 3 \theta = 3 \ln 2 - 2 \geq 0$, Theorem~\ref{thm:enhanced-bound} follows.

\end{document}